\documentclass[a4paper,UKenglish]{lipics-v2018}

\pdfoutput=1 

\usepackage{microtype}

\newcommand{\polylog}{\mathrm{PolyLogTime}}
\newcommand{\npolylog}{\mathrm{NPolyLogTime}}

\title{The Polylog-Time Hierarchy Captured by Restricted Second-Order Logic}

\author{Flavio Ferrarotti}{Software Competence Center Hagenberg, {Hagenberg, Austria}}{flavio.ferrarotti@scch.at}{https://orcid.org/0000-0003-2278-8233}{}

\author{Sen\'{e}n Gonz\'{a}lez}{Software Competence Center Hagenberg, {Hagenberg, Austria}}{senen.gonzalez@scch.at}{}{}

\author{Klaus-Dieter Schewe}{Christian-Doppler Laboratory for Client-Centric Cloud Computing, {Linz, Austria}}{kdschewe@acm.org}{}{}

\author{Jos\'{e} Mar\'{i}a Turull-Torres}{Universidad Nacional de La Matanza, {Buenos Aires, Argentina}}{jturull@unlam.edu.ar}{}{}

\authorrunning{F. Ferrarotti, S. Gonz\'{a}les, K.-D. Schewe and J. M. Turull-Torres}

\Copyright{Flavio Ferrarotti, Sen\'{e}n Gonz\'{a}lez, K.-D. Schewe and Jos\'{e} Mar\'{i}a Turull-Torres}

\subjclass{Theory of computation $\rightarrow$ Finite Model Theory}


\category{}

\funding{The research reported in this paper results from the project {\em Higher-Order Logics and Structures} supported by the Austrian Science Fund (FWF: \textbf{[I2420-N31]}).}

\EventEditors{John Q. Open and Joan R. Access}
\EventNoEds{2}
\EventLongTitle{42nd Conference on Very Important Topics (CVIT 2016)}
\EventShortTitle{CVIT 2016}
\EventAcronym{CVIT}
\EventYear{2016}
\EventDate{December 24--27, 2016}
\EventLocation{Little Whinging, United Kingdom}
\EventLogo{}
\SeriesVolume{42}
\ArticleNo{23}
\nolinenumbers 

\allowdisplaybreaks

\begin{document}

\maketitle

\begin{abstract}
Let $\mathrm{SO}^{\mathit{plog}}$ denote the restriction of second-order logic, where second-order quantification ranges over relations of size at most poly-logarithmic in the size of the structure. In this article we investigate the problem, which Turing machine complexity class is captured by Boolean queries over ordered relational structures that can be expressed in this logic. For this we define a hierarchy of fragments $\Sigma^{\mathit{plog}}_m$ (and $\Pi^{\mathit{plog}}_m$) defined by formulae with alternating blocks of existential and universal second-order quantifiers in quantifier-prenex normal form. We first show that the existential fragment $\Sigma^{\mathit{plog}}_1$ captures $\npolylog$, i.e. the class of Boolean queries that can be accepted by a non-deterministic Turing machine with random access to the input in time $O((\log n)^k)$ for some $k \ge 0$. Using alternating Turing machines with random access input allows us to characterise also the fragments $\Sigma^{\mathit{plog}}_m$ (and $\Pi^{\mathit{plog}}_m$) as those Boolean queries with at most $m$ alternating blocks of second-order quantifiers that are accepted by an alternating Turing machine. Consequently, $\mathrm{SO}^{\mathit{plog}}$ captures the whole poly-logarithmic time hierarchy.  We demonstrate the relevance of this logic and complexity class by several problems in database theory.
 \end{abstract}
 
\section{Introduction}

According to Immerman, the credo of descriptive complexity theory is that ``the computational complexity of all problems in Computer Science can be understood via the complexity of their logical descriptions'' \cite[p.5]{Immerman99}. Starting from Fagin's fundamental result \cite{fagin:1973} that the existential fragment $\mathrm{SO}\exists$ of second-order logic over finite relational structures captures all decision problems that are accepted by a non-deterministic Turing machine in polynomial time---in other words: $\mathrm{SO}\exists$ captures the complexity class $\mathrm{NP}$---many more connections between logics and Turing complexity classes have been discovered (see e.g. the monographs by Immerman \cite{Immerman99} and Libkin \cite{Libkin04} or the collection \cite{graedel:eatcs2007}). The polynomial time hierarchy is captured by full second-order logic $\mathrm{SO}$ over finite relational structures \cite{Stockmeyer76}, but it is unknown, whether there exists a logic capturing the complexity class $\mathrm{P}$.

In order to understand the gap between $\mathrm{P}$ and $\mathrm{NP}$ a lot of research has been dedicated to extensions of first-order logic. For instance, adding transitive closure to first-order logic defines the logic $\mathrm{FO}[\mathrm{TC}]$, which captures NLOGSPACE \cite[p.150]{Immerman99}. Blass, Gurevich and Shelah investigate the choiceless fragment of $\mathrm{P}$ \cite{blass:apal1999} using Abstract State Machines \cite{boerger:2003}. They present a logic that expresses all properties expressible in any other $\mathrm{P}$ logic in the literature, but does does not capture all of $\mathrm{P}$.

The project {\em Higher-Order Logics and Structures} is dedicated to a somehow inverse approach, the investigation of semantically restricted higher-order logics over finite structures and their relationship to Turing complexity. The logic $\mathrm{SO}^\omega$ introduced by Dawar in \cite{Dawar98} and the related logic $\mathrm{SO}^F$ introduced in \cite{GrossoT10}, respectively, provide the main background for the theoretical line of work in this direction. Both logics restrict the interpretation of second-order quantifiers to relations closed under equivalence of types of the tuples in the given relational structure. Through the study of different semantic restrictions over the existential second-order logic with second-order quantification restricted to binary relations, many interesting results regarding the properties of the class of problems expressible in this logic (known as binary NP) were established \cite{DurandLS98}. Another relevant example of a semantic restriction over existential second-order logic can be found in \cite{LautemannST94}. 

The expressive power of higher-order logics (beyond second-order) on finite structures has been studied, among a few others, by Kuper and Vardi \cite{KuperV88}, by Leivant \cite{Leivant89} and by Hull and Su \cite{HullS91}. However, the \emph{exact} characterization of each prenex fragment of higher-order logics (in terms of oracle machines) over finite structures is more recent and it is due to Hella and Turull Torres \cite{HellaT03,HellaT06}. Independently, Kolodziejczyk \cite{Kol04,Kol05} characterized the prenex fragments of higher-order logic in terms of alternating Turing machines, but taking also into account the arity of the higher-order variables.
Starting from studies about the expressiveness of restricted higher-order logics in \cite{FerrarottiRT13} and fragments on higher-order logics that collapse to second-order \cite{ferrarotti:corr2016,ferrarotti:wollic2017}, all of which defining complexity classes that include $\mathrm{NP}$, the question comes up, which restrictions to $\mathrm{SO}$ give rise to meaningful complexity classes.

\subsection{Main Contributions}

In this paper we concentrate on complexity classes inside POLYLOG-SPACE. Analogous to the polynomial time hierarchy inside PSPACE we investigate a poly-logarithmic time hierarchy $\mathrm{PLH}$, where $\tilde{\Sigma}_1^{\mathit{plog}}$ is defined by $\npolylog$ capturing all decision problems that can be accepted by a non-deterministic Turing machine in time $O((\log n)^k)$ for some $k \ge 0$, where $n$ is the size of the input. In order to be able to deal with the sublinear time constraint random access to the input is assumed. Higher complexity classes $\tilde{\Sigma}_m^{\mathit{plog}}$ (and $\tilde{\Pi}_m^{\mathit{plog}}$) in the hierarchy are defined analogously using alternating Turing machines with a bound $m$ on the alternations.

In the same spirit we define the logic $\mathrm{SO}^{\mathit{plog}}$, which denotes the restriction of second-order logic, where second-order quantification ranges over relations of size at most poly-logarithmic in the size of the structure. A hierarchy of fragments $\Sigma^{\mathit{plog}}_m$ (and $\Pi^{\mathit{plog}}_m$) is then defined by formulae with alternating blocks of existential and universal second-order quantifiers in quantifier-prenex normal form. We first show that the existential fragment $\Sigma^{\mathit{plog}}_1$ captures $\npolylog = \tilde{\Sigma}^{\mathit{plog}}_1$, i.e. the class of Boolean queries that can be accepted by a non-deterministic Turing machine with random access to the input in time $O((\log n)^k)$ for some $k \ge 0$. Using alternating Turing machines with random access input allows us to characterise also the fragments $\Sigma^{\mathit{plog}}_m$ (and $\Pi^{\mathit{plog}}_m$) as those Boolean queries with at most $m$ alternating blocks of second-order quantifiers that are accepted by an alternating Turing machine. That is, we obtain $\Sigma^{\mathit{plog}}_m = \tilde{\Sigma}^{\mathit{plog}}_m$ (and $\Pi^{\mathit{plog}}_m = \tilde{\Pi}^{\mathit{plog}}_m$). Consequently, $\mathrm{SO}^{\mathit{plog}}$ captures the whole poly-logarithmic time hierarchy $\mathrm{PLH}$.
 
\subsection{Related Work}

The logic $\mathrm{SO}^{\mathit{plog}}$ is similar to the restricted second-order logic (let us denote it as $\mathrm{SO}^r$) defined by David A. Mix Barrington in~\cite{Barr92}. He uses $\mathrm{SO}^r$ to characterize a class of families of circuits $\mathit{qAC}^0$, showing\footnote{The result in \cite{Barr92} is actually more general, allowing any set of Boolean functions  $\mathcal{F}$ of $n^{{O}(1)}$ inputs complying with a padding property and containing the functions $\mathrm{OR}$ and $\mathrm{AND}$. The restricted second-order logic is defined by extending first-order logic with a second-order quantifier $Q_{f}$ for each $f \in \mathcal{F}$ which range over relations on the sub-domain $\{1, \ldots, \log n\}$, where $n$ is the size of the interpreting structure. The case related to our result is when $\mathcal{F} = \{\mathrm{OR}, \mathrm{AND}\}$, which gives raise to restricted existential and universal second-order quantifiers.} that the class of Boolean queries computable by $\mathrm{DTIME}[(\log n)^{{O}(1)}]$ $\mathrm{DCL}$-uniform families of Boolean circuits of unbounded fan-in, size $2^{({\log n})^{{O}(1)}}$ and depth ${O}(1)$, coincides with the class of Boolean queries expressible in $\mathrm{SO}^{r}$. 

There is a well known result (\cite{Immerman99}, Theorem~5.22) which shows that the class of first-order uniform families of Boolean circuits of unbounded fan-in, size $n^{{O}(1)}$ and depth ${O}(1)$, coincides with the class of languages $\mathrm{ATIME}[\log n, {O}(1)]$ that are accepted by random-access alternating Turing machines that make at most $\log n$ steps and at most ${O}(1)$ alternations between existential and universal states. The intuitive idea is that as alternating Turing machines have bounded fan-out in their computation trees, to implement an AND (OR) gate of unbounded fan-in, a full balanced tree of depth logarithmic in the size of the circuits, of universal (existential) states is needed. Then it appears as natural that $qAC^0$ coincides with the whole poly-logarithmic time hierarchy $\mathrm{PLH}$ as defined in this paper, since $\mathrm{PLH} = \mathrm{ATIME}[(\log n)^{{O}(1)}, {O}(1)]$ and $(\log n)^{{O}(1)}$ is the logarithm of the size $2^{(\log n)^{{O}(1)}}$ of the circuits in $qAC^0$.   

Then the fact that $\mathrm{SO}^{\mathit{plog}}$ captures the whole class $\mathrm{PLH}$ could also be seen as a corollary of Barrington's theorem in~\cite{Barr92} (see Section~3, page~89). This however does \emph{not} applies to our main results, i.e., the capture of $\npolylog$ by the existential fragment of $\mathrm{SO}^{\mathit{plog}}$ and the one-to-one correspondence between the prenex fragments of $\mathrm{SO}^{\mathit{plog}}$ and the corresponding levels of $\mathrm{PLH}$. The critical difference between Barrington's $\mathrm{SO}^r$ logic and $\mathrm{SO}^{\mathit{plog}}$ is that we impose a restriction in the first-order logic sub-formulae, so that the universal first-order quantifier is only allowed to range over sub-domains of polylog size. This is a key feature since otherwise the fist-order sub-formulae of the $\Sigma^{\mathit{plog}}_m$ (and $\Pi^{\mathit{plog}}_m$) fragments of $\mathrm{SO}^{\mathit{plog}}$ would need at least linear time to be evaluated. Of course, Barrington does not need to define such constraint because he always speaks of the whole class $\mathrm{PLH}$, and we show indeed that for every first-order logic formula there is an equivalent $\mathrm{SO}^{\mathit{plog}}$ formula. 

It is \emph{not} a trivial task to establish an exact correspondence between the levels $\tilde{\Sigma}_m^{\mathit{plog}}$ of the polylog-time hierarchy PLH and sub-classes of families of circuits in $\mathit{qAC}^0$. On this regard, we have explored the relationship with the sub-classes $\mathit{qAC}^0_m$ of families of circuits in $\mathit{qAC}^0$ where the path from an input gate to the output gate with the maximum number of alternated gates of unbounded fan-in of type AND and OR in the circuits is $m$. Let us denote as $\exists qAC^0_{m}$ ($\forall qAC^0_{m}$) the sub-class of $qAC^0_{m}$ where the output gate is of type $\mathrm{OR}$ ($\mathrm{AND}$). We were able to show that:
\begin{enumerate}[i.]
\item $\exists qAC^0_{m}$ $\subseteq$ $\tilde{\Sigma}^{\mathit{plog}}_m$, and $\forall qAC^0_{m}$ $\subseteq$ $\tilde{\Sigma}^{\mathit{plog}}_{m+1}$ for all $m \geq 1$.
\item If $t,k \ge 1$, $\psi \in {\Sigma}^{1,\mathit{plog}}_t$, the first-order sub-formula $\varphi$ of $\psi$ belongs to ${\Sigma}^{0}_k$ and the vocabulary of $\psi$ includes the $\mathrm{BIT}$ predicate, then there is a family $\mathcal{C}_{\psi}$ of Boolean circuits in $\exists qAC^0_{t+k}$ that computes the Boolean query expressed by $\psi$.
\end{enumerate}
Sketches of the proofs of~(i) and~(ii) are included in Appendix~\ref{hirearchy_in_aAC0} (see Lemmas~\ref{a} and~\ref{b}, respectively). Whether the converse of these results hold or not is still open. 

\subsection{Organization}

We reach our results by following an inductive itinerary. After presenting some preliminaries in Section \ref{sec:preliminaries}, we introduce the logic $\mathrm{SO}^{\mathit{plog}}$ in Section~\ref{sec:soplog}. We do this in a comprehensive way, following our research program in the line of restricting sizes or other properties of valuating relations in higher-order logics. We give examples of natural queries expressible in $\mathrm{SO}^{\mathit{plog}}$, such as the classes $\mathrm{DNFSAT}$ of satisfiable propositional formulas in disjunctive normal form and $\mathrm{CNFTAUT}$ of propositional tautologies in conjunctive normal form, both defined in as early as 1971 (\cite{Cook_71}). The expression of such queries in $\mathrm{SO}^{\mathit{plog}}$ can in general be done by means of relatively simple and elegant formulae, despite the restriction we impose in the universal first-order quantification. This is not fortuitous, but the consequence of using in the definition of $\mathrm{SO}^{\mathit{plog}}$ a more relaxed notion of second-order quantification than that used by Barrington in the definition of $\mathrm{SO}^r$. Indeed, the second-order quantifiers in $\mathrm{SO}^{\mathit{plog}}$ range over arbitrary relations of polylog size on the number of elements of the domain, not just over relations defined on the set formed by the first $\log n$ elements of that domain as in $\mathrm{SO}^r$. The descriptive complexity of $\mathrm{SO}^{\mathit{plog}}$ is not increased by this more liberal definition of polylog restricted second-order quantifiers. The fragments $\Sigma^{\mathit{plog}}_m$ and $\Pi^{\mathit{plog}}_m$ of formulae in quantifier-prenex normal form are defined using the classical approach in second-order logic, showing that every $\mathrm{SO}^{\mathit{plog}}$ formula can be written in this normal form. This forms the basis for the definition of the hierarchy inside $\mathrm{SO}^{\mathit{plog}}$. 

Section~\ref{sec:arithexampls} shows how the logic captures the poly-logarithmically bounded binary arithmetics necessary to prove our main results.

In Section~\ref{sec:plh} we define a non-deterministic Turing machine inspired in the random access deterministic and alternating Turing machines used in \cite{barrington:jcss1990}, as well as the corresponding complexity class $\mathrm{NPolyLogTime}$ and the levels of the implied polylog-time hierarchy $\mathrm{PLH}$ inside POLYLOG-SPACE as indicated above. 

Section \ref{sec:main} contains our main results. First we give a detailed, constructive proof of the fact that the existential fragment of $\mathrm{SO}^{\mathit{plog}}$, i.e. $\Sigma^{\mathit{plog}}_1$, captures the complexity class $\npolylog$. This part is the most challenging since, as pointed out earlier, it requires the use of a restricted form of first-order universal quantification. 
After that, we follow the inductive path and prove the expressive power of the fragments $\Sigma^{\mathit{plog}}_m$ and $\Pi^{\mathit{plog}}_m$, for every $m \geq 1$, and prove that each layer is characterized by a random-access alternating Turing machine with polylog time and $m$ alternations. The fact that $\mathrm{PLH} = \mathrm{ATIME}[(\log n)^{{O}(1)}, {O}(1)] = \mathrm{SO}^{\mathit{plog}}$ follows as a simple corollary.

We conclude the paper with a brief summary and outlook in Section \ref{sec:schluss}.

\section{Preliminaries}\label{sec:preliminaries}

Unless otherwise stated, we work with ordered finite structures and assume that all vocabularies include the relation and constant symbols: $\leq$, $\mathrm{SUCC}$, $\mathrm{BIT}$, $0$, $1$, $\mathit{logn}$ and $\mathit{max}$. In every structure $\bf A$, $\leq$ is interpreted as a total ordering of the domain $A$ and $\mathrm{SUCC}$ is interpreted by the successor relation corresponding to the $\leq^{\bf A}$ ordering. The constant symbols $0$, $1$ and $\mathit{max}$ are in turn interpreted as the minimum, second and maximum elements under the $\leq^{\bf A}$ ordering and the constant $\mathit{logn}$ as $\left\lceil \log_2 |A| \right\rceil$. By passing to an isomorphic copy, we assume that $A$ is the set $\{0, 1, \ldots, n-1\}$ of natural numbers less than $n$, where $n$ is the cardinality $|A|$ of $A$. Then $\mathrm{BIT}$ is interpreted by the following binary relation:
\[\mathrm{BIT}^{\bf A} = \{(i, j) \in A^2 \mid \text{Bit $j$ in the binary representation of $i$ is $1$}\}.\]
We assume that all structures have at least \emph{two} elements. This results in a cleaner presentation, avoiding the trivial case of structures with only one element which would satisfy $0 = 1$. In this paper, $\log n$ always refers to the binary logarithm of $n$, i.e. $\log_2 n$. We write $\log^k n$ as a shorthand for $(\left\lceil\log n \right\rceil)^k$ and finally $\mathit{\log n-1}$ as $z$ such as $SUCC(z,\mathit{logn})$. 
 
\section{\texorpdfstring{$\mathrm{SO}^{\mathit{plog}}$}{TEXT}: A Restricted Second-Order Logic}{\label{sec:soplog}}

We define $\mathrm{SO}^{\mathit{plog}}$ as the restricted second-order logic obtained by extending \emph{existential} first-order logic with (1) universal and existential second-order quantifiers that are restricted to range over relations of poly-logarithmic size in the size of the structure, and (2) universal first-order quantifiers that are restricted to range over the tuples of such poly-logarithmic size relations. 

\begin{definition}[Syntax of $\mathrm{SO}^{\mathit{plog}}$]
firstFor every $r{\geq}1$ and $k{\geq}0$, the language of $\mathrm{SO}^{\mathit{plog}}$ extends the language of first-order logic with countably many second-order variables $X_1^{r,\log^k}$, $X_2^{r,\log^k}, \dots$ of {\em arity $r$} and {\em exponent $k$}. The set of well-formed $\mathrm{SO}^{\mathit{plog}}$-formulae (wff) of vocabulary $\sigma$ is inductively defined as follows:
\begin{enumerate}[i.]

\item Every well-formed formula of vocabulary $\sigma$ in the existential fragment of first-order logic with equality is a wff.

\item If $X^{r,\log^k}$ is a second-order variable and $t_1, \ldots, t_r$ are first-order terms, then both $X^{r,\log^k}(t_1, \ldots, t_r)$ and $\neg X^{r,\log^k}(t_1, \ldots, t_r)$ are wff.

\item If $\varphi$ and $\psi$ are wff, then $(\varphi \wedge \psi)$ and $(\varphi \vee \psi)$ are wff.
    
\item If $\varphi$ is a wff, $X^{r,\log^k}$ is a second-order variable and $\bar{x}$ is an $r$-tuple of first-order variables, then $\forall \bar{x} (X^{r,\log^k}(\bar{x}) \rightarrow \varphi)$ is a wff.    
    
\item If $\varphi$ is a wff and $x$ is a first-order variables, then $\exists x \varphi$ is a wff.   
     
\item If $\varphi$ is an $\mathrm{SO}^{\mathit{plog}}$-formula and $X^{r,\log^k}$ is a second-order variable, then both $\exists X^{r,\log^k} \varphi$ and $\forall X^{r,\log^k} \varphi$ are wff's.
    
\end{enumerate}
\end{definition}

Note that the first-order terms $t_i$ in these rules are either first-order variables $x_1, x_2, \ldots$ or constant symbols; we do not consider function symbols. Whenever the arity is clear from the context, we write $X^{\log^k}$ instead of $X^{r,\log^k}$. 
    
\begin{definition}[Semantics of $\mathrm{SO}^{\mathit{plog}}$]
Let $\mathbf{A}$ be a $\sigma$-structure where $|A| = n \geq 2$. A valuation over $\mathbf{A}$ is any function \textit{val} which assigns appropriate values to all first- and second-order variables and satisfies the following constraints: 
\begin{itemize}

\item If $x$ is a first-order variable then $\mathit{val}(x) \in A$. 

\item If $X^{r,\log^k}$ is a second-order variable, then 
\[\mathit{val}(X^{r,\log^k}) \in \{R\subseteq A^r \mid |R| \leq (\lceil \log n \rceil)^k\}.\]
\end{itemize}
\end{definition}
As usual, we say that a valuation $\mathit{val}$ is $V$-equivalent to a valuation $\mathit{val}'$ if $\mathit{val}(V') = \mathit{val}'(V')$ for all variables $V'$ other than $V$.

$\mathrm{SO}^{\mathit{plog}}$ extends the notion of satisfaction of first-order logic, with the following rules:
\begin{itemize}
     \item $\mathbf{A},\mathit{val} \models X^{r,\log^k}(x_1,\dots,x_r) $ iff $(\mathit{val}(x_1),\dots,\mathit{val}(x_r))\in \mathit{val}( X^{r,\log^k})$.
     \item $\mathbf{A},\mathit{val} \models \neg X^{r,\log^k}(x_1,\dots,x_r) $ iff $(\mathit{val}(x_1),\dots,\mathit{val}(x_r))\not\in \mathit{val}( X^{r,\log^k})$.
     
     \item $\mathbf{A},\mathit{val} \models \exists X^{r,\log^k} (\varphi)$  iff there is a valuation $\mathit{val}'$ which is $X^{r,\log^k}$-equivalent to $\mathit{val}$ such that $\mathbf{A}, \mathit{val}' \models \varphi$.
     
     \item $\mathbf{A},\mathit{val} \models \forall X^{r,\log^k} (\varphi)$  iff, for all valuations $\mathit{val}'$ which are $X^{r,\log^k}$-equivalent to $\mathit{val}$, it holds that $\mathbf{A}, \mathit{val}' \models \varphi$.
     
\end{itemize}

\begin{remark}\label{r1}
The standard (unbounded) universal quantification of first-order logic formulae of the form 
$\forall x \varphi$ can be expressed in $\mathrm{SO}^{\mathit{plog}}$ by formulae of the form $\forall X^{\log^0} \forall x (X^{\log^0}(x) \rightarrow \varphi)$. Thus, even though $\mathrm{SO}^{\mathit{plog}}$ only allows a restricted form of universal first-order quantification, it can nevertheless express every first-order query.  
\end{remark} 

We denote by $\Sigma^{\mathit{plog}}_m$, where $m \geq 1$, the class of $\mathrm{SO}^{\mathit{plog}}$-formulae of the form:
\[\exists X^{\log^{k_{11}}}_{11} \cdots \exists X^{\log^{k_{1s_1}}}_{1s_1} \forall X^{\log^{k_{21}}}_{21} \cdots \forall X^{\log^{k_{2s_2}}}_{2s_2} \cdots Q X^{\log^{k_{m1}}}_{m1} \cdots Q X^{\log^{k_{ms_m}}}_{ms_m} \psi,\]
where $Q$ is either $\exists$ or $\forall$ depending on whether $m$ odd or even, respectively, and $\psi$ is an $\mathrm{SO}^{\mathit{plog}}$-formula \emph{free} of second-order quantifiers.  
Analogously, we denote by $\Pi^{\mathit{plog}}_m$ the class of $\mathrm{SO}^{\mathit{plog}}$-formulae of the form:
\[\forall X^{\log^{k_{11}}}_{11} \cdots \forall X^{\log^{k_{1s_1}}}_{1s_1} \exists X^{\log^{k_{21}}}_{21} \cdots \exists X^{\log^{k_{2s_2}}}_{2s_2} \cdots Q X^{\log^{k_{m1}}}_{m1} \cdots Q X^{\log^{k_{ms_m}}}_{ms_m} \psi.\]

We say that an $\mathrm{SO}^{\mathit{plog}}$-formula is in \emph{Skolem normal form} (SNF aka \emph{quantifier prenex normal form}) if it belongs to either $\Sigma^{\mathit{plog}}_m$ or $\Pi^{\mathit{plog}}_m$ for some $m \geq 1$.

\begin{lemma}\label{lem-snf}

For every $\mathrm{SO}^{\mathit{plog}}$-formula $\varphi$, there is an equivalent $\mathrm{SO}^{\mathit{plog}}$-formula $\varphi'$ that is in SNF.

\end{lemma} 

The straightforward proof is given in Appendix \ref{app-snf}.

\subsection{Examples of Queries in \texorpdfstring{$\mathrm{SO}^{\mathit{plog}}$}{TEXT}}\label{sec:examples2}

We start with a simple, but useful example. Let $X$ and $Y$ be $\mathrm{SO}^{\mathit{plog}}$ variables of the form $X^{r_1,\log^{k}}$ and $Y^{r_2,\log^{k}}$. The following $\Sigma^{\mathit{plog}}_1$ formula, denoted as $|X| \leq |Y|$, express that the cardinality of (the relation assigned by the current valuation of) $X$ is less than or equal to that of $Y$. \\[0.1cm]
\hspace*{0.38cm}$\exists R \Big(\forall \bar{x} \big(X(\bar{x}) \to \exists \bar{y} \big( Y(\bar{y})\wedge R(\bar{x},\bar{y}) \wedge \forall \bar{z} (X(\bar{z})\to (\bar{z}{\neq} \bar{x} \to \neg R(\bar{z},\bar{y})))\big)\big)\Big)$,\\[0.1cm]
where $R$ is an $\mathrm{SO}^{\mathit{plog}}$ variable of arity $r_1+r_2$ and exponent $k$. We write $|X|{=}|Y|$ to denote $|X|{\leq}|Y| \wedge |Y|{\leq}|X|$.

Let $G=(V,E)$ be an $n$-node undirected graph. The following sentence expresses a poly-logarithmically bounded version of the clique NP-complete problem. It holds iff $G$ contains a clique of size $\lceil \log n\rceil^k$. 
\[\exists I S \big(\mathrm{DEF}_k(I) \wedge |S|{=}|I| \wedge\forall x \big(S(x) {\to} (V(x) \wedge \forall y (S(y) \to (x\neq y \to (E(x,y) \wedge E(y,x)))))\big)\big)\]
Other bounded versions of classical Boolean NP-complete problems that are easily expressible in $\Sigma^{\mathit{plog}}_1$ are for instance to decide whether $G$ has an induced subgraph of size $\lceil \log n \rceil^k$ that is $3$-colourable, or whether a $G$ has an induced subgraph which is isomorphic to another given graph of at most polylog size w.r.t. the size of $G$.

\begin{remark}
Although possible, it is much more complex and cumbersome to write these examples of queries in Barrington's restricted second-order logic $\mathrm{SO}^r$. For instance we cannot directly express, as we do in the $\mathrm{SO}^{\mathit{plog}}$-formula that defines the poly-logarithmically bounded version of the clique problem, that there is a set $S$ of arbitrary nodes of $G$ (where $S$ is of size $\lceil \log n\rceil^k$) such that the sub-graph induced by $S$ in $G$ is a clique. To do that in $\mathrm{SO}^r$ we would need to define instead a set of arbitrary binary numbers, which would need to be encoded into a relation of arity $k+2$ defined on the sub-domain $\{1, \ldots, \log n\}$, and then use $\mathrm{BIT}$ to check whether the nodes of $G$ corresponding to these binary numbers induce a sub-graph of $G$ which is a clique.      
\end{remark}

The class DNFSAT of satisfiable propositional formulas in disjunctive normal form is an example of a natural problem expressible in $\mathrm{SO}^{\mathit{plog}}$.
In the standard encoding of DNF formulae as word models of alphabet $\sigma = \{(,),\wedge,\vee, \neg, 0, 1, X\}$, DNFSAT is decidable in $P$~\cite{Cook_71}. In this encoding, 
the input formula is a disjunction of arbitrarily many clauses enclosed in pairs of matching parenthesis. Each clause is the conjunction of an arbitrary number of literals. Each literal is a variable of the form $X_w$, where $w \in \{0, 1\}^*$, possibly preceded by a negation symbol.
Obviously, the complement NODNFSAT of DNFSAT is also in P. In $\Pi^{\mathrm{plog}}_2$ NODNFSAT can be defined by a sentence stating that for every clause there is a pair of complementary literals. Every clause is logically defined by a pair of matching parentheses such that there is  no parenthesis in between. A pair of complementary literals is defined by a bijection (of size $< \log n$) between the substrings that encode two literals, which preserve the binary numbers and such that exactly one of them is negated. The actual formula $\Pi^{\mathrm{plog}}_2$ is included in Appendix~\ref{appendix:NODNFSAT}. 
Similarly, DNFSAT can be defined $\Sigma^{\mathit{plog}}_2$ by a sentence stating that there is a clause that does not have a pair of complementary literals.



\section{Bounded Binary Arithmetic Operations in \texorpdfstring{$\Sigma^{\mathit{plog}}_1$}{TEXT}}\label{sec:arithexampls}\label{sec:examples}

We define $\Sigma^{\mathit{plog}}_1$-formulae that describe the basic (bounded) arithmetic operations of sum, multiplication, division and modulo among binary positive integers between $0$ and $2^{\lceil\log n \rceil^k}-1$ for some fixed $k \geq 1$. These formulae are latter required for proving our main result regarding the expressive power of the prenex fragments of $\mathrm{SO}^{\mathit{plog}}$.  

Without loss of generality we assume that the size $n$ of the structures in which the formulae are evaluated is at least $3$. This simplifies the presentation avoiding the special case in which $\lceil \log n \rceil = 1$. In our approach, binary numbers between $0$ and $\lceil 2^{(\log n)^k} \rceil - 1$ are represented by means of ($\mathrm{SO}^{\mathit{plog}}$) relations. 
\begin{definition}\label{def:repBinaryNumbers}
Let $b = b_0 \cdots b_l$ be a binary number, where $b_0$ and $b_l$ are the least and most significant bits of $b$, respectively, and $l \leq \lceil \log n \rceil^k$. Let $B = \{0, \ldots, \lceil \log n \rceil -1\}$. The relation $R_b$ \emph{encodes the binary number} $b$ if the following holds: $(a_0, \ldots, a_{k-1}, a_k) \in R_b$ iff $(a_0, \ldots, a_{k-1}) \in B^k$ is the $i$-th tuple in the increasing numerical order (numbers read left to right) of $B^k$, $a_k = 0$ if $i > l$, and $a_k = b_i$ if $0 \leq i \leq l$. 
\end{definition}

Note that the size of $R_b$ is exactly $\lceil \log n \rceil^k$, and thus $R_b$ is a valid valuation for $\mathrm{SO}^{\mathit{plog}}$ variables of the form $X^{k+1, \log^k}$.
The numerical order relation $\leq_k$ among $k$-tuples can be defined as follows: 
\begin{align}
\bar{x} \leq_1 \bar{y} \quad &\equiv \quad x_0 \leq y_0 \quad \text{and} \notag\\
\bar{x} \leq_k \bar{y} \quad &\equiv \quad (x_0 \leq y_0 \wedge x_0 \neq y_0) \vee (x_0 = y_0 \wedge (x_1, \ldots, x_{k-1}) \leq_{k-1} (y_1, \ldots, y_{k-1})) \label{arith1}
\end{align}

In our approach, we need a successor relation $\mathrm{SUCC}_k$ among the $k$-tuples in $B^k$, where $B$ is the set of integers between $0$ and $\lceil \log n \rceil-1$ (cf. Definition~\ref{def:repBinaryNumbers}). 
\begin{align}
  \mathrm{SUCC}_1(\bar{x},\bar{y}) &\equiv y_0 \leq \mathit{logn} \wedge y_0 \neq \mathit{logn} \wedge \mathrm{SUCC}(x_0,y_0) \quad \text{and} \notag\\
\mathrm{SUCC}_k(\bar{x},\bar{y}) &\equiv y_0 \leq \mathit{logn} \wedge y_0 \neq \mathit{logn} \wedge \notag \\
[\, (y_0 = x_0 \wedge & \mathrm{SUCC}_{k-1}(x_1,\dots,x_{k-1},y_1,\dots,y_{k-1})) \, \vee (\mathrm{SUCC}(x_0, y_0) \wedge \notag\\
\mathrm{SUCC} (x_1, & \mathit{logn}) \wedge \cdots \wedge \mathrm{SUCC}(x_{k-1},\mathit{logn}) \wedge y_1= 0 \wedge \cdots \wedge y_{k-1}=0)\, \label{arith2}]
\end{align}

It is useful to define an auxiliary predicate $\mathrm{DEF_k}(I)$, where $I$ is a second-order variable of arity and exponent $k$, such that ${\bf A}, \mathit{val} \models \mathrm{DEF_k}(I)$ if $\mathit{val}(I) = B^k$. Please, note that we abuse the notation, writing for instance $\bar{x} = \bar{0}$ instead of $x_0 = 0 \wedge \cdots \wedge x_{k-1} = 0$. Such abuses of notation should nevertheless be clear from the context.   
\begin{align}
\mathrm{DEF}_k(I) & \equiv \exists \bar{x} ( \bar{x} = \bar{0} \wedge I(\bar{x})) \wedge \forall \bar{y} (I(\bar{y}) \rightarrow ((\mathrm{SUCC}(y_0, \mathit{logn}) \wedge \cdots \notag \\
&  \wedge  \mathrm{SUCC}(y_k, \mathit{logn})) \vee \exists \bar{z} (SUCC_k(\bar{y},\bar{z}) \wedge I(\bar{z})))) \label{arith3}
\end{align}

The formula $\mathrm{BIN}_k(X)$, where $X$ is a second-order variable of arity $k+1$ and exponent $k$, expresses that $X$ encodes (as per Definition~\ref{def:repBinaryNumbers}) a binary number between $0$ and $2^{\lceil\log n \rceil^k} - 1$. 
\begin{equation}
\mathrm{BIN}_k(X) \equiv \exists I (\mathrm{DEF}_k(I) \wedge \forall \bar{x} (I(\bar{x}) \to (X(\bar{x},0) \vee X(\bar{x},1)))) \label{arith4}
\end{equation}

As $X$ is of exponent $k$, the semantics of $\mathrm{SO}^{\mathit{plog}}$ determines that the number of tuples in any valid valuation of $X$ is always bounded by $\lceil\log n \rceil^k$. It is then clear that the previous formula is equivalent to  
\[ \exists I (\mathrm{DEF}_k(I) \wedge \forall \bar{x} (I(\bar{x}) \to ((X(\bar{x},0) \wedge \neg X(\bar{x},1)) \vee (X(\bar{x},1) \wedge \neg X(\bar{x},0))))).\]

In the following, $\mathrm{BIN}_k(X, I)$ denotes the sub-formula $\forall \bar{x} (I(\bar{x}) \to (X(\bar{x},0) \vee X(\bar{x},1)))$ of $\mathrm{BIN}_k(X)$. 

The comparison relations $X =_k Y$ and $X <_k Y$ ($X$ is strictly smaller than $Y$) among binary numbers encoded as second-order relations are defined as follows:
\begin{equation}
X =_k Y \equiv \exists I \big(\mathrm{DEF}_k(I) \wedge \mathrm{BIN}_k(X, I) \wedge \mathrm{BIN}_k(Y,I) \wedge {=_k}(X, Y, I) \big), \label{arith5}
\end{equation}

\noindent
where ${=_k}(X, Y, I) \equiv \forall \bar{x} \big( I(\bar{x}) \to \exists z (X(\bar{x},z) \wedge Y(\bar{x}, z))\big)$.
\begin{equation}
X <_k Y \equiv \exists I \big(\mathrm{DEF}_k(I) \wedge \mathrm{BIN}_k(X, I) \wedge \mathrm{BIN}_k(Y,I) \wedge {<_k}(X, Y, I)\big) , \label{arith6} 
\end{equation}

\noindent
where ${<_k}(X, Y, I) \equiv \exists \bar{x} \big( I(\bar{x}) \wedge X(\bar{x}, 0) \wedge Y(\bar{x}, 1) \land \forall \bar{y} \big(I(\bar{y}) \to (\bar{y} \leq_k \bar{x} \vee \exists z (X(\bar{y},z) \wedge Y(\bar{y}, z)))\big)\big)$.

Sometimes we need to determine if the binary number encoded in (the current valuation of) a second-order variable $X$ of arity $k+1$ and exponent $k$ corresponds to the binary representation of an individual $x$ from the domain. The following $\mathrm{BNUM}_k(X,x)$ formula holds whenever that is the case.  
\begin{align}
    \mathrm{BNUM}_k&(X,x) \equiv 
\exists I \big(\mathrm{DEF}_k(I) \wedge \mathrm{BIN}_k(X,I) \wedge \notag \\
&    \forall \bar{y} \big(I(\bar{y}) \to \big((y_0 = 0 \wedge \cdots \wedge y_{k-2} = 0 \wedge (X(\bar{y}, 1) \leftrightarrow BIT(x, y_{k-1}))) \vee \notag \\
&     \hspace*{1.9cm}(\neg (y_0 = 0 \wedge \cdots \wedge y_{k-2} = 0) \wedge X(\bar{y},0))\big)\big)\big) \label{arith7}
\end{align}

We use $\mathrm{BNUM}_k(X,x,I)$ to denote the sub-formula $\forall \bar{y} (I(\bar{y}) \to ((y_0 = 0 \wedge \cdots \wedge y_{k-2} = 0 \wedge (X(\bar{y}, 1) \leftrightarrow BIT(x, y_{k-1}))) \vee (\neg (y_0 = 0 \wedge \cdots \wedge y_{k-2} = 0) \wedge X(\bar{y},0))))$ of $\mathrm{BNUM}_k(X,x)$.

We now proceed to define $\Sigma^{\mathit{plog}}_1$-formulae that describe basic (bounded) arithmetic operations among binary numbers. We start with $\mathrm{BSUM}_k(X,Y,Z)$, where $X$, $Y$ and $Z$ are free-variables of arity $k+1$ and exponent $k$. This formula holds if (the current valuation of) $X$, $Y$ and $Z$ represent binary numbers between $0$ and $2^{\lceil\log n \rceil^k} - 1$, and $X + Y = Z$. The second-order variables $I$ and $W$ in the formula are of arity $k$ and $k+1$, respectively, and both have exponent $k$. We use the traditional carry method, bookkeeping the carried digits in $W$. 
\begin{align}
\mathrm{BSUM}_k&(X,Y,Z) \equiv \notag \\
\exists I  W \big(&\mathrm{DEF}_k(I) \wedge \mathrm{BIN}_k(X,I) \wedge \mathrm{BIN}_k(Y,I) \wedge \mathrm{BIN}_k(Z,I) \wedge \mathrm{BIN}_k(W,I) \wedge \notag \\
&W(\bar{0},0) \wedge ({<_k}(X, Z, I) \vee {=_k}(X, Z, I)) \wedge ({<_k}(Y, Z, I) \vee {=_k}(Y, Z, I)) \wedge \notag \\
&\forall \bar{x} (I(\bar{x}) \to ( (\bar{x} = \bar{0} \wedge \varphi(X, Y, Z)) \vee  \notag\\
&\hspace*{1.9cm} (\exists \bar{y} (SUCC_k(\bar{y},\bar{x}) \wedge \psi(\bar{x}, \bar{y}, W, X, Y)) \wedge  \alpha(\bar{x}, W, X, Y, Z))))\big) \label{arith8}
\end{align}
where $\varphi$ holds if the value of the first bit of $Z$ is consistent with the sum of the first bits of $X$ and $Y$. $\psi$ holds if the value of the bit in position $\bar{x}$ of $W$ is consistent with the sum of the values of the bits in the position preceding $\bar{x}$ of $W,X$ and $Y$. Finally, $\alpha$ holds if the value of the bit in position $\bar{x}$ of $Z$ is consistent with the sum of the corresponding bit values of $W,X$ and $Z$. For reference, see the actual $\varphi$, $\psi$ and $\alpha$ formulae in Appendix~\ref{app-arithmetics}.

For the operation of (bounded) multiplication of binary numbers, we define a formula $\mathrm{BMULT}_k(X, Y, Z)$, where $X$, $Y$ and $Z$ are free-variables of arity $k+1$ and exponent $k$. This formula holds if (the current valuations of) $X$, $Y$ and $Z$ represent binary numbers between $0$ and $2^{\lceil\log n \rceil^k} - 1$, and $X \cdot Y = Z$. 

The strategy to express the multiplication consists on keeping track of the (partial) sums of the partial products by means of a relation $R \subset B^k \times B^k \times \{0,1\}$ of size $\lceil \log n \rceil^{2k}$ (recall that $B = \{0, \ldots, \lceil \log n \rceil -1\}$). We take $X$ to be the multiplicand and $Y$ to be the multiplier. Let $\bar{a} \in B^k$ be the $i$-th tuple in the numerical order of $B^k$, let $R|_{\bar{a}}$ denote the restriction of $R$ to those tuples starting with $\bar{a}$, i.e., $R|_{\bar{a}} = \{(\bar{b}, c) \mid (\bar{a}, \bar{b}, c) \in R\}$, and let $\mathit{pred}(\bar{a})$ denote the immediate predecessor of $\bar{a}$ in the numerical order of $B^k$, then the following holds:

\begin{enumerate}[a.]
\item If $\bar{a} = \bar{0}$ and $Y(\bar{a},0)$, then $R|_{\bar{a}}$ encodes the binary number $0$.
\item If $\bar{a} = \bar{0}$ and $Y(\bar{a},1)$, then $R|_{\bar{a}} = X$.
\item If $\bar{a} \neq \bar{0}$ and $Y(\bar{a},0)$, then $R|_{\bar{a}} = R|_{\mathit{pred}(\bar{a})}$.
\item If $\bar{a} \neq \bar{0}$ and $Y(\bar{a},1)$, then (the binary number encoded by) $R|_{\bar{a}}$ results from adding $R|_{\mathit{pred}(\bar{a})}$ to the $(i-1)$-bits arithmetic left-shift of $X$.
\end{enumerate}

$\mathrm{BMULT}_k(X, Y, Z)$ holds if $Z = R|_{(a_0, \ldots, a_{k-1})}$ for $a_0 = \cdots = a_{k-1} = \lceil \log n \rceil -1$. 
Following this strategy, it is not difficult to write the actual formula for $\mathrm{BMULT}_k(X, Y, Z)$. See formula~(\ref{arith9}) in Appendix~\ref{app-arithmetics}.

The operations of division and modulo are expressed by $\mathrm{BDIV}_k(X, Y, Z, M)$, where $X$, $Y$, $Z$ and $M$ are free-variables of arity $k+1$ and exponent $k$. This formula holds if $Z$ is the quotient and $M$ the modulo (remainder) of the euclidean division of $X$ by $Y$, i.e., if  $Y\cdot Z + M = X$. See formula~(\ref{arith11}) in Appendix~\ref{app-arithmetics}.


\section{The Poly-logarithmic Time Hierarchy}\label{sec:plh}

The sequential access that Turing machines have to their tapes makes it impossible to compute anything in sub-linear time. Therefore, logarithmic time complexity classes are usually studied using models of computation that have random access to their input. As this also applies to the poly-logarithmic complexity classes studied in this paper, we adopt a Turing machine model that has a \emph{random access} read-only input, similar to the log-time Turing machine in~\cite{barrington:jcss1990}.

A \emph{random-access Turing machine} is a multi-tape Turing machine with (1) a read-only (random access) \emph{input} of length $n+1$, (2) a fixed number of read-write \emph{working tapes}, and (3) a read-write input \emph{address-tape} of length $\lceil \log n \rceil$.

Every cell of the input as well as every cell of the address-tape contains either $0$ or $1$ with the only exception of the ($n+1$)st cell of the input, which is assumed to contain the endmark $\triangleleft$. In each step the binary number in the address-tape either defines the cell of the input that is read or if this number exceeds $n$, then the ($n+1$)st cell containing $\triangleleft$ is read.   

\begin{example}\label{ex:machine}
Let polylogCNFSAT be the class of satisfiable propositional formulae in conjunctive normal form with $c \leq \lceil \log n \rceil^k$ clauses, where $n$ is the length of the formula. Note that the formulae in polylogCNFSAT tend to have few clauses and many literals. We define a random-access Turing machine $M$ which decides polylogCNFSAT. The alphabet of $M$ is $\{0,1,\#,+,-\}$. The input formula is encoded in the input tape as a list of $c \leq \lceil \log n \rceil^k$ indices, each index being a binary number of length $\lceil \log n \rceil$, followed by $c$ clauses. For every $1 \leq i \leq c$, the $i$-th index points to the first position in the $i$-th clause. Clauses start with $\#$ and are followed by a list of literals. Positive literals start with a $+$, negative with a $-$. The $+$ or $-$ symbol of a literal is followed by the ID of the variable in binary. $M$ proceeds as follows: (1) Using binary search with the aid of the ``out of range'' response $\triangleleft$, compute $n$ and $\lceil \log n \rceil$. (2) Copy the indices to a working tape, counting the number of indices (clauses) $c$. (3) Non-deterministically guess $c$ input addresses $a_1, \ldots, a_c$, i.e., guess $c$ binary numbers of length $\lceil \log n \rceil$. (4) Using $c$ $1$-bit flags, check that each $a_1, \ldots, a_c$ address falls in the range of a different clause. (5) Check that each $a_1, \ldots, a_c$ address points to an input symbol $+$ or $-$. (6) Copy the literals pointed by $a_1, \ldots, a_c$ to a working tape, checking that there are \emph{no} complementary literals. (7) Accept if all checks hold.
\end{example}

Let $L$ be a language accepted by a random-access Turing machine $M$. Assume that for some function $f$ on the natural numbers, $M$ makes at most $O(f(n))$ steps before accepting an input of length $n$. If $M$ is deterministic, then we write $L \in \mathrm{DTIME}[f(n)]$. If $M$ is non-deterministic, then we write $L \in \mathrm{NTIME}[f(n)]$. We define the classes of deterministic and non-deterministic poly-logarithmic time computable problems as follows:
\[ \polylog = \bigcup_{k \in \mathbb{N}} \mathrm{DTIME}[\log^k n] \qquad \, \npolylog = \bigcup_{k \in \mathbb{N}} \mathrm{NTIME}[\log^k n] \]
The non-deterministic random-access Turing machine in Example~\ref{ex:machine} clearly works in polylog-time. Therefore, polylogCNFSAT $\in \npolylog$.

In order to relate our logic $\mathrm{SO}^{\mathit{plog}}$ to these Turing complexity classes we adhere to the usual conventions concerning a binary encoding of finite structures~\cite{Immerman99}. Let $\sigma = \{R^{r_1}_1, \ldots, R^{r_p}_p, c_1, \ldots, c_q\}$ be a vocabulary, and let ${\bf A}$ with $A = \{0, 1, \ldots, n-1\}$ be an ordered structure of vocabulary $\sigma$. Each relation $R_i^{\bf A} \subseteq A^{r_i}$ of $\bf A$ is encoded as a binary string $\mathrm{bin}(R^{\bf A}_i)$ of length $n^{r_i}$ where $1$ in a given position indicates that the corresponding tuple is in $R_i^{\textbf{A}}$.
Likewise, each constant number $c^{\bf A}_j$ is encoded as a binary string $\mathrm{bin}(c^{\bf A}_j)$ of length $\lceil \log n \rceil$. The encoding of the whole structure $\mathrm{bin}(\textbf{A})$ is simply the concatenation of the binary strings encodings its relations and constants: 
\[\mathrm{bin}(\textbf{A}) = \mathrm{bin}(R_1^{\textbf{A}}) \cdots \mathrm{bin}(R_p^{\textbf{A}}) \cdot \mathrm{bin}(c^{\bf A}_1) \cdots \mathrm{bin}(c^{\bf A}_q).\] 

The length $\hat{n} = |\mathrm{bin}(\textbf{A})|$ of this string is $n^{r_1}+\cdots+n^{r_p} + q \lceil \log n \rceil$, where $n = |A|$ denotes the size of the input structure ${\bf A}$. Note that $\log \hat{n} \in O(\lceil \log n \rceil)$, so $\mathrm{NTIME}[\log^k \hat{n}] = \mathrm{NTIME}[\log^k n]$ (analogously for $\mathrm{DTIME}$). Therefore, we will consider random-access Turing machines, where the input is the encoding $\mathrm{bin}(\textbf{A})$ of the structure \textbf{A} followed by the endmark $\triangleleft$.

In the present work our machine is also based in an alternating Turing machine. An alternating Turing machine comes with a set of states $Q$ that is partitioned into subset $Q_\exists$ and $Q_\forall$ of so-called existential and universal states. Then a configuration $c$ is accepting iff
\begin{itemize}

\item $c$ is in a final accepting state,

\item $c$ is in an existential state and there exists a next accepting configuration, or

\item $c$ is in a universal state, there exists a next configuration and all next configurations are accepting.

\end{itemize}

In analogy to our definition above we can define a \emph{random-access alternating Turing machine}. The languages accepted by such a machine $M$, which starts in an existential state and makes at most $O(f(n))$ steps before accepting an input of length $n$ with at most $m$ alternations between existential and universal states, define the complexity class $\mathrm{ATIME}[f(n),m]$. Analogously, we define the complexity class $\mathrm{ATIME}^{op}[f(n),m]$ comprising languages that are accepted by a random-access alternating Turing machine that starts in a universal state and makes at most $O(f(n))$ steps before accepting an input of length $n$ with at most $m$ alternations between universal and existential states. With this we define
\[ \tilde{\Sigma}_m^{\mathit{plog}} = \bigcup_{k \in \mathbb{N}} \mathrm{ATIME}[\log^k n,m] \quad \text{and} \quad \tilde{\Pi}_m^{\mathit{plog}} = \bigcup_{k \in \mathbb{N}} \mathrm{ATIME}^{op}[\log^k n,m] . \]

The poly-logarithmic time hierarchy is then defined as $\mathrm{PLH} = \bigcup_{m \ge 1} \tilde{\Sigma}_m^{\mathit{plog}}$. Note that $\tilde{\Sigma}_1^{\mathit{plog}} = \npolylog$ holds. 

\begin{remark}

Note that a simulation of a $\npolylog$ Turing machine $M$ by a deterministic machine $N$ requires checking all computations in the tree of computations of $M$. As $M$ works in time $({\log n})^{O(1)}$, $N$ requires time $2^{{\log n}^{O(1)}}$. This implies $\npolylog \subseteq \mathrm{DTIME}(2^{{\log n}^{O(1)}})$, which is the complexity class called quasipolynomial time of the fastest known algorithm for graph isomorphism \cite{babai:stoc2016}, which further equals 
$\mathrm{DTIME}({n^{{\log n}^{O(1)}}})$\footnote{This relationship appears quite natural in view of the well known relationship $\mathrm{NP} = \mathrm{NTIME}(n^{O(1)}) \subseteq \mathrm{DTIME}(2^{{n}^{O(1)}}) = \mathrm{EXPTIME}$.}.

\end{remark}

\section{Expressive power of the Quantifier-Prenex Fragments of \texorpdfstring{$\mathrm{SO}^{\mathit{plog}}$}{TEXT}} \label{sec:main}

We say that a logic $\mathcal{L}$ captures the complexity class $\mathcal{K}$ iff the following holds:
\begin{itemize} 

\item For every $\mathcal{L}$-sentence $\varphi$ the language $\{\mathrm{bin}({\bf A}) \mid {\bf A} \models \varphi \}$ is in $\mathcal{K}$, and

\item For every property $\mathcal{P}$ of (binary encodings of) structures that can be decided with complexity in $\mathcal{K}$, there is a sentence $\varphi_{\mathcal{P}}$ of $\mathcal{L}$ such that ${\bf A} \models \varphi_{\mathcal{P}}$ iff $\bf A$ has the property $\mathcal{P}$. 

\end{itemize} 

Recall that we only consider ordered, finite structures $\mathbf{A}$. Our main result is that $\mathrm{SO}^{\mathit{plog}}$ captures PLH, which we will prove in this section.

\begin{theorem}\label{pedsoplog}

Over ordered structures with sucessor relation, $\mathrm{BIT}$ and constants for $\log n$, the minimum, second and maximum elements, $\Sigma^{\mathit{plog}}_1$ captures $\npolylog$.

\end{theorem}

\begin{proof}
\ \textbf{Part a.} We first show $\Sigma^{\mathit{plog}}_1 \subseteq NPolyLogTime$, i.e. a non-deterministic random access Turing Machine \textbf{M} can evaluate every sentence $\phi$ in $\Sigma^{\mathit{plog}}_1$ in poly-logarithmic time. 

Let $\phi = \exists X_1^{r_1,\log^{k_1}} \dots \exists X_m^{r_m,\log^{k_m}} \varphi$, where $\varphi$ is a first-order formula with the restrictions given in the definition of $\mathrm{SO}^{\mathit{plog}}$. Given a $\sigma$-structure \textbf{A} with $|dom(\textbf{A})|=n$, \textbf{M} first guesses values for $X_1^{r_1,\log^{k_1}},\dots, X_m^{r_m,\log^k_m}$ and then checks if $\varphi$ holds. As $val(X_i^{r_i,\log^{k_i}})$ is a relation of arity $r_i$ with at most $\log^{k_i} n$ tuples, \textbf{M} has to guess $r_i*{(\log(n))}^{k_i}$ values in $A$, each of which encoded in $\left \lceil log(n)\right \rceil$ bits. Thus, the machine has to generate $E = \sum\limits_{i=1}^{m}\left ( r_i*{(\log(n))}^{k_i+1}\right )$ bits in total. As $E \in O(\lceil \log n \rceil^{k_{\max} + 1})$, the generation of the values $val(X_i^{r_i,\log^{k_i}})$ requires time in $O(\lceil \log n \rceil^{k^\prime})$ for some $k^\prime$.
    
In order to check the validity of $\varphi$ we distinguish two cases: (1) $\varphi = \exists x \psi$, and (2) $\varphi = \forall \bar{x}(X^{r,\log^{k}}(\bar{x})\to \psi)$.

\ \\ \textbf{Case (1).}
Here \textbf{M} first guesses $x$, for which at most $\lceil \log n \rceil$ steps are required. Then we get the following cases for $\psi$:
\begin{itemize}

\item If $\psi$ is a first-order quantifier-free formula, then according to the proof of \cite[Theorem 5.30]{Immerman99} $\varphi$ can be checked in $O(log(n))$ time. Thus, checking $\varphi$ can be done in poly-logarithmic time.
        
\item If $\psi = \psi_1 \vee \psi_2$ (or $\psi = \psi_1 \wedge \psi_2$), then \textbf{M} has to check if $\psi_1$ or $\psi_2$ (or both, respectively) holds, which requires at most the time for checking both $\psi_1$ and $\psi_2$. Thus, by induction the checking of $\varphi$ can be done in poly-logarithmic time.
        
\item If $\psi$ is not an first-order quantifier-free formula, then by induction the checking of $\psi$ can be done in poly-logarithmic time, hence also the checking of $\varphi$.

\end{itemize}

\textbf{Case (2).}
As \textbf{M} has already guessed a value for $X^{r,\log^{k}}$, it remains to check $\{ \bar{a} / \bar{x} \} . \psi$ for every element $\bar{a}$ in this relation. The number of such tuples is in $O(\log^{k^\prime} n)$, and we get the following cases:
\begin{itemize}
    
\item If $\psi$ is a first-order quantifier-free formula, a disjunction or a conjunction, then we can use the same argument as in case (1).
        
\item If $\psi$ is not an first-order quantifier-free formula, then by induction each check $\{ \bar{a} / \bar{x} \} . \psi$ can be done in poly-logarithmic time, say in $O(\log^{\ell} n)$, hence also the checking of $\varphi$ is done in $O(\log^{k^\prime + \ell} n)$ time.
        
\end{itemize}

\noindent
\textbf{Part b.} Next we show $\npolylog  \subseteq \Sigma^{\mathit{plog}}_1$. For this let \textbf{M} be a non-deterministic random access Turing Machine that accepts a $\sigma$-structure \textbf{A} in $O(\log^k n)$ steps. We may assume that the input is encoded by the bitstring $\mathrm{bin}(\textbf{A})$ of length $\left\lceil log(\hat{n}) \right\rceil$. Furthermore, let the set of states be $Q = \{ q_0,\dots,q_f \}$, where $q_0$ is the initial state, $q_f$ is the only final state, and in the initial state the header of the tapes are in position 0, the working tape is empty and the index-tape is filled with zeros.

As \textbf{M} runs in time $\left\lceil \log n\right\rceil^k$, it visits at most $\left\lceil \log n\right\rceil^k$ cells in the working tape. Thus, we can model positions on the working tape and time by $k$-tuples $\bar{p}$ and $\bar{t}$, respectively. Analogously, the length of the index tape is bound by $\left\lceil \log n\right\rceil^{k'} $, so we can model the positions in the index tape by $k'$-tuples $\bar{d}$. We use auxiliary relations $I$ and $I^\prime$ to capture $k$-tuples and $k^\prime$-tuples, respectively, over $\{ 0 ,\dots, \lceil \log n \rceil \}$.  We define those relations using $\mathrm{DEF_k}(I)$ and $\mathrm{DEF_{k'}}(I')$ in the same way as in (\ref{arith3}).
As \textbf{M} works non-deterministically, it makes a choice in every step. Without loss of generality we can assume that the choices are always binary, which we capture by a relation $C$ of arity $k + k^\prime + 1$; $C(\bar{t},\bar{d},c)$ expresses that at time $\bar{t}$ any position $\bar{d}$ in the index-tape has the value $c \in \{0,1\}$, which denotes the two choices.

In order to construct a sentence in $\Sigma^{\mathit{plog}}_1$ that is satisfied by the structure \textbf{A} iff the input $\mathrm{bin}(\textbf{A})$ is accepted by \textbf{M} we first describe logically the operation of the random access Turing machine \textbf{M}, then exploit the acceptance of $\mathrm{bin}(\textbf{A})$ for at least one computation path.

We use predicates $T_0, T_1, T_2$, where $T_i(\bar{t},\bar{p})$ indicates that at time $\bar{t}$ the working tape at position $\bar{p}$ contains $i$ for $i \in \{ 0,1 \}$ and the blank symbol for $i=2$, respectively. The following formulae express that the working tape is initially empty, and at any time a cell can only contain one of the three possible symbols:
\begin{align}
& \forall \bar{p} \; I(\bar{p}) \rightarrow T_2(\bar{0},\bar{p}) \notag \\
& \forall \bar{t} I(\bar{t}) \rightarrow \forall \bar{p} \; I(\bar{p}) \rightarrow ( T_0(\bar{t},\bar{p}) \rightarrow \neg T_1(\bar{t},\bar{p}) \wedge \neg T_2(\bar{t},\bar{p}) ) \notag \\
& \forall \bar{t} I(\bar{t}) \rightarrow \forall \bar{p} \; I(\bar{p}) \rightarrow ( T_1(\bar{t},\bar{p}) \rightarrow \neg T_0(\bar{t},\bar{p}) \wedge \neg T_2(\bar{t},\bar{p}) ) \notag \\
& \forall \bar{t} I(\bar{t}) \rightarrow \forall \bar{p} \; I(\bar{p}) \rightarrow ( T_2(\bar{t},\bar{p}) \rightarrow \neg T_0(\bar{t},\bar{p}) \wedge \neg T_1(\bar{t},\bar{p}) ) \label{eq-2}
\end{align}

Then we use a predicate $H$ with $H(\bar{t},\bar{p})$ expressing that at time $\bar{t}$ the header of the working tape is in position $\bar{p}$. This gives rise to the following formulae:
\begin{gather}
H(\bar{0},\bar{0}) \qquad\qquad\qquad \forall \bar{t} I(\bar{t}) \rightarrow \exists \bar{p} ( I(\bar{p}) \wedge H(\bar{t},\bar{p}) ) \notag \\
\forall \bar{t} I(\bar{t}) \rightarrow \forall \bar{p} \; I(\bar{p}) \rightarrow ( H(\bar{t},\bar{p}) \rightarrow \forall \bar{p}^\prime ( I(\bar{p}^\prime) \rightarrow H(\bar{t},\bar{p}^\prime) \rightarrow \bar{p} = \bar{p}^\prime ) ) \label{eq-3}
\end{gather}

Predicates $S_i$ for $i = 1 ,\dots, f$ are used to express that at time $\bar{t}$ the machine \textbf{M} is in the state $q_i \in Q$, which using (\ref{arith1}) gives rise to the formulae
\begin{gather}
S_0(\bar{0}) \wedge \forall \bar{t}\big( I(\bar{t}) \rightarrow ( \bar{t} \neq \bar{0} \rightarrow \bigvee_{0 \le i \le f} S_i(\bar{t}) )\big) \wedge 
\exists \bar{t}_f \Big( \forall \bar{t} \big( I(\bar{t}) \rightarrow \bar{t} \le_k \bar{t}_f \big)  \wedge S_f(\bar{t}_f)\Big) \notag \\
\bigwedge_{0 \le i \le f} \; \forall \bar{t} \Big(I(\bar{t}) \rightarrow \big(S_i(\bar{t}) \rightarrow \bigwedge_{0 \le j \le f, j \neq i} \neg S_j(\bar{t})\big)\Big) \label{eq-4}
\end{gather}

The following formulae exploiting (\ref{arith2}) describe the behaviour of \textbf{M} moving in every step its header either to the right, to the left or not at all (which actually depends on the value for $c$ in $C(\bar{t},\bar{d},c)$):
\begin{align}
& \forall \bar{t} I(\bar{t}) \wedge \bar{t} \neq \bar{0} \wedge H(\bar{t},\bar{0}) \rightarrow
\exists \bar{t}^\prime, \bar{d}^\prime \; ( \mathrm{SUCC}_k(\bar{t}^\prime,\bar{t}) \wedge \notag \\
& \hspace*{3cm} \mathrm{SUCC}_k(\bar{0},\bar{d}^\prime) \wedge ( H(\bar{t}^\prime,\bar{0}) \vee H(\bar{t}^\prime,\bar{d}^\prime) ) ) \notag \\
& \forall \bar{t}, \bar{d} I(\bar{t}) \wedge \bar{t} \neq \bar{0} \wedge I^\prime(\bar{d}) \wedge H(\bar{t},\bar{d}) \rightarrow ( \bar{d} \neq \bar{0} \wedge \bar{d} \neq \mathit{last} \rightarrow \notag \\
& \hspace*{1.5cm} \exists \bar{t}^\prime, \bar{d}_1 , \bar{d}_2 \; ( \mathrm{SUCC}_k(\bar{t}^\prime,\bar{t}) \wedge \mathrm{SUCC}_k(\bar{d}_1,\bar{d}) \wedge \mathrm{SUCC}_k(\bar{d},\bar{d}_2) \wedge \notag \\
& \hspace*{3cm}  ( H(\bar{t}^\prime,\bar{d}_1) \vee H(\bar{t}^\prime,\bar{d}) \vee H(\bar{t}^\prime,\bar{d}_2) ) ) \notag \\
& \forall \bar{t} I(\bar{t}) \wedge \bar{t} \neq \bar{0} \wedge H(\bar{t},\mathit{last}) \rightarrow
\exists \bar{t}^\prime, \bar{d}^\prime \; ( \mathrm{SUCC}_k(\bar{t}^\prime,\bar{t}) \wedge \notag \\
& \hspace*{3cm} \mathrm{SUCC}_k(\bar{d}^\prime,\mathit{last}) \wedge ( H(\bar{t}^\prime,\mathit{last}) \vee H(\bar{t}^\prime,\bar{d}^\prime) ) ) \label{eq-5}
\end{align}
    
Furthermore, we use predicates $L_i$ ($i \in \{ 0, 1, 2 \}$) to describe that \textbf{M} reads at time $\bar{t}$ the value $i$ (for $i \in \{ 0, 1 \}$) or $\triangleleft$ for $i=2$, respectively. As exactly one of these values is read, we obtain the following formulae:
\begin{alignat}{2}
& \forall \bar{t} \; I(\bar{t}) \rightarrow ( L_0(\bar{t}) \vee L_1(\bar{t}) \vee L_2(\bar{t}) ) \quad &
& \forall \bar{t} \; I(\bar{t}) \rightarrow ( L_0(\bar{t}) \rightarrow \neg L_1(\bar{t}) \wedge \neg L_2(\bar{t}) ) \notag \\
& \forall \bar{t} \; I(\bar{t}) \rightarrow ( L_1(\bar{t}) \rightarrow \neg L_0(\bar{t}) \wedge \neg L_2(\bar{t}) ) \quad &
& \forall \bar{t} \; I(\bar{t}) \rightarrow ( L_2(\bar{t}) \rightarrow \neg L_0(\bar{t}) \wedge \neg L_1(\bar{t}) ) \label{eq-6}
\end{alignat}

The conjunction of the formulae in (\ref{eq-2})-(\ref{eq-6}) with all second-order variables existentially quantified merely describes the operation of the Turing machine \textbf{M}. If \textbf{M} accepts the input $\mathrm{bin}(\textbf{A})$ for at least one computation path, i.e. for one sequence of choices, we can assume without loss of generality that if at time $\bar{t}$ with $\bar{d}$ on the index-tape the bit $c$ indicating the choice equals the value read from the input, then this will lead to acceptance. Therefore, in order to complete the construction of the required formula in $\Sigma^{\mathit{plog}}_1$ we need to express this condition in our logic.

The bit \textbf{M} reads from the input corresponds to the binary encoding of the relations and constants in the structure $\mathbf{A}$. In order to detect, which tuple or which constant is actually read, we require several auxiliary predicates. We use predicates $M_i$ ($i=0, \dots, k^\prime$) to represent the numbers $n^i$, which leads to the formulae
\begin{gather}
\mathrm{BNUM}_{k^\prime}(M_0,1,I^\prime),\;\mathrm{BNUM}_{k^\prime}(M_1,\max,I^\prime) \;\text{and }
 \mathrm{BMULT}_{k^\prime}(M_1,M_{i-1},M_i) \;\text{for}\; i \ge 2 \label{eq-7}
\end{gather}

For this we exploit the definitions in (\ref{arith7}) and (\ref{arith9}). Note that the latter one is a formula that is not in SNF. In Section \ref{sec:examples} following (\ref{arith9}) we already showed how to turn such a formula into a formula in $\Sigma^{\mathit{plog}}_1$ in SNF, which requires additional auxiliary second-order variables. The same hold for several of the following formulae.

Next we use relations $P_i$ ($i=0,\dots,p+1$) representing the position in $\mathrm{bin}(\textbf{A})$, where the encoding of $R_{i+1}^{\mathbf{A}}$ for the relation $R_i$ starts (for $0 \le i \le p-1$), the encoding of where the constants $c_j$ ($j=1,\dots,q$) starts (for $i=p$), and finally representing the length of $\mathrm{bin}(\textbf{A})$ (for $i=p+1$). As each constant requires $\lceil \log n \rceil$ bits we further use auxiliary relations $N_i$ (for $i=1,\dots,q$) to represent $i \cdot \lceil \log n \rceil$, which we require to detect, which constant is read. This leads to the following formulae (exploiting (\ref{arith7}) and (\ref{arith8})):
\begin{align}
& \mathrm{BNUM}_{k^\prime}(P_0,0,I^\prime) \; , \; \bigwedge_{1 \le i \le p} \mathrm{BSUM}_{k^\prime}(P_{i-1},M_{r_i},P_i) \; \text{and} \; \mathrm{BSUM}_{k^\prime}(P_p,N_q,P_{p+1}) \notag \\
& \mathrm{BNUM}_{k^\prime}(N_1,\mathit{logn},I^\prime) \quad \text{and} \; \bigwedge_{1 \le i \le q} \mathrm{BSUM}_{k^\prime}(N_{i-1},N_1,N_i) \label{eq-8}
\end{align}

Finally, we can express the acceptance condition linking the relation $C$ to the input $\mathrm{bin}(\textbf{A})$. In order to ease the representation we use for fixed $\bar{t}$ the shortcut $C_{\bar{t}}$ with $C_{\bar{t}}(\bar{d},c) \leftrightarrow C(\bar{t},\bar{d},c)$. Likewise we use shortcuts with subscript $\bar{t}$ for additional auxiliary predicates $D_i$ ($i=0,\dots,p$), $Q_i$, $Q_i^\prime$, $Q_i^{\prime\prime}$ and $Q_i^{\prime\prime\prime}$ ($i=1,\dots,r_{max}$) which we need for arithmetic operations on the length of $\mathrm{bin}(\textbf{A})$, which is represented by $P_{p+1}$. We also use $\le_{k^\prime}(X,Y,I)$ as shortcut for $<_{k^\prime}(X,Y,I) \vee =_{k^\prime}(X,Y,I)$ defined in~(\ref{arith5}) and~(\ref{arith6}).

For fixed $\bar{t}$ with $I(\bar{t})$ the relation $C_{\bar{t}}$ represents a position in the bitstring $\mathrm{bin}(\textbf{A})$, which is either at the end, within the substring encoding the constants $c_j^{\mathbf{A}}$, or within the substring encoding the relation $R_i^{\mathbf{A}}$. The following three formulae (using (\ref{arith6}), (\ref{arith7}), (\ref{arith8}), and (\ref{arith11})) with fixed $\bar{t}$ correspond to these cases:
\begin{align}
    & <_{k^\prime}(P_{p+1},C_{\bar{t}},I^\prime) \rightarrow L_2(\bar{t}) \notag \\
    & <_{k^\prime}(P_p,C_{\bar{t}},I^\prime) \wedge \le_{k^\prime}(C_{\bar{t}},P_{p+1},I^\prime) \wedge \mathrm{BSUM}_{k^\prime}(P_p,D_{0,\bar{t}},C_{\bar{t}}) \wedge \notag \\
      & \mathrm{BDIV}_{k^\prime}(D_{0,\bar{t}}, N_1, Q_{1,\bar{t}}, Q_{1,\bar{t}}^\prime)\to\exists xy \big( \mathrm{BNUM}_{k^\prime}(Q_{1,\bar{t}},x) \wedge \mathrm{BNUM}_{k^\prime}(Q_{1,\bar{t}}^\prime,y)
      \wedge \notag\\
      & \hspace*{5.6cm}\mathrm{BIT}(c_x,y) {\leftrightarrow} L_1(\bar{t})\big)
    \notag\\
    & \bigwedge_{1 \le i \le p} <_{k^\prime}(P_{i-1},C_{\bar{t}},I^\prime) \wedge \le_{k^\prime}(C_{\bar{t}},P_i,I^\prime) \wedge \mathrm{BSUM}_{k^\prime}(P_{i-1},D_{i,\bar{t}},C_{\bar{t}}) \rightarrow \notag \\
    & \quad \exists \bar{x} \Bigg( \bigwedge_{1 \le j \le r_i} \Big( \mathrm{BNUM}_{k^\prime}(Q_{j,\bar{t}}^{\prime\prime\prime},x_j)
    \wedge \mathrm{BDIV}_{k^\prime}(D_{i,\bar{t}}, M_j, Q_{j,\bar{t}}, Q_{j,\bar{t}}^\prime)\notag\\
    & \quad \quad \quad\quad \quad \quad\quad \quad \quad\quad \quad \quad\wedge \mathrm{BDIV}_{k^\prime}(Q_{j,\bar{t}}, M_1, Q_{j,\bar{t}}^{\prime\prime}, Q_{j,\bar{t}}^{\prime\prime\prime}) \Big)\wedge
    \notag\\
    & \quad\quad\quad\quad \Big(\big(
    L_1(\bar{t}) ) \to  R_i(\bar{x})
    \big)\vee\big(
    L_0(\bar{t}) ) \to  \neg R_i(\bar{x})
    \big) \Big) \Bigg)
    \label{eq-9}
\end{align}
Note that in the second case $Q_{1,\bar{t}}$ represents an index $j \in \{ 1,\dots,q \}$ and $Q_{1,\bar{t}}^\prime$ represents the read bit of the constant $c_j^{\mathbf{A}}$ in $\mathrm{bin}(\textbf{A})$. In the third case $D_{i,\bar{t}}$ represents the read position $d$ in the encoding on relation $R_i^{\mathbf{A}}$, which represents a particular tuple, for which we use $Q_{j,\bar{t}}^{\prime\prime\prime}$ to determine every value of the tuple and depending of the read in $L_i(\bar{t})$ check if that particular tuple is in the relation or not.

Finally, the sentence $\Psi$ describing acceptance by \textbf{M} results from building the conjunction of the formulae in (\ref{eq-2})-(\ref{eq-9}), expanding the macros as shown in Section \ref{sec:examples}, which brings in additional second-order variables, and existentially quantifying all second-order variables. Due to our construction we have $\mathbf{A} \models \Psi$ iff $\mathbf{A}$ is accepted by \textbf{M}.\end{proof}
\begin{theorem}\label{pedsoplog2}

Over ordered structures with sucessor relation, $\mathrm{BIT}$ and constants for $\log n$, the minimum, second and maximum elements, $\Pi^{\mathit{plog}}_1$ captures $\tilde{\Pi}^{\mathit{plog}}_1$.

\end{theorem}
\begin{theorem}\label{pedsoplog3}

Over ordered structures with sucessor relation, $\mathrm{BIT}$ and constants for $\log n$, the minimum, second and maximum elements, $\Sigma^{\mathit{plog}}_m$ captures $\tilde{\Sigma}^{\mathit{plog}}_m$ and $\Pi^{\mathit{plog}}_m$ captures $\tilde{\Pi}^{\mathit{plog}}_m$ for all $m \ge 1$.

\end{theorem}
Sketches of the proofs of Theorems \ref{pedsoplog2} and \ref{pedsoplog3} will be given in Appendix \ref{app-plh}. With Theorem \ref{pedsoplog3} the following corollary is a straightforward consequence of the definitions.
\begin{corollary}

Over ordered structures with sucessor relation, $\mathrm{BIT}$ and constants for $\log n$, the minimum, second and maximum elements, $\mathrm{SO}^{\mathit{plog}}$ captures the poly-logarithmic time hierarchy PLH.

\end{corollary}

\section{Conclusions}\label{sec:schluss}
We investigated $\mathrm{SO}^{\mathit{plog}}$, a restriction of second-order logic, where second-order quantification ranges over relations of poly-logarithmic size and first-order quantification is restricted to the existential fragment of first-order logic plus universal quantification over variables in the scope of a second-order variable. In this logic we defined the poly-logarithmic hierarchy PLH using fragments $\Sigma^{\mathit{plog}}_m$ (and $\Pi^{\mathit{plog}}_m$) defined by formulae with alternating blocks of existential and universal second-order quantifiers in quantifier-prenex normal form. We showed that the existential fragment $\Sigma^{\mathit{plog}}_1$ captures $\npolylog$, i.e. the class of Boolean queries that can be accepted by a non-deterministic Turing machine with random access to the input in time $O(\log^k n)$ for some $k \ge 0$. In general, $\Sigma^{\mathit{plog}}_m$ captures the class of Boolean queries that can be accepted by an alternating Turing machine with random access to the input in time $O(\log^k n)$ for some $k \ge 0$ with at most $m$ alternations between existential and universal states. Thus, PLH is captured by $\mathrm{SO}^{\mathit{plog}}$. 

For the proofs the restriction of first-order quantification is essential, but it implies that we do not have closure under negation. As a consequence we do not have a characterisation of the classes $\text{co-}\Sigma^{\mathit{plog}}_m$ and $\text{co-}\Pi^{\mathit{plog}}_m$. These consitute open problems. Furthermore, PLH resides in the complexity class PolyLogSpace, which is known to be different from $\mathrm{P}$, but it is conjectured that PolyLogSpace and $\mathrm{P}$ are incomparable. Whether the inclusion of PLH in PolyLogSpace is strict is another open problem.

The theory developed in this article and its proofs make intensive use of alternating Turing machines with random access to the input. We observe that it appears awkward to talk about poly-logarithmic time complexity, when actually an unbounded number of computation branches have to be exploited in parallel. It appears more natural to refer directly to a computation model that involves directly unbounded parallelism such as Abstract State Machines that have already been explored in connection with the investigation of choiceless polynomial time \cite{blass:apal1999}. We also observe that a lot of the technical difficulties in the proofs result from the binary encodings that are required in order to make logical structures accessible for Turing machines. The question is, whether a different, more abstract treatment would help to simplify the technically complicated proofs. These more general questions provide further invitations for future research.


\bibliography{SOpolylog}

\begin{appendix}

\section*{Appendix}

\section{The Hierarchy in \texorpdfstring{$qAC^0$}{TEXT}}\label{hirearchy_in_aAC0}

We assume from the reader a basic knowledge of Circuit Complexity (\cite{Immerman99} is a good reference for the subject). We consider a circuit as a connected acyclic digraph with arbitrary many input  nodes and exactly one output node.
 As in \cite{Barr92} we define  $qAC^0$ as the class of $\mathrm{DTIME} [\log^{O(1)} n]\;$ $\;\mathrm{DCL}$ uniform families of Boolean circuits of unbounded fan-in, Size $2^{\log^{O(1)} n}$ and Depth $O(1)$.
If $\mathcal{C}$ is a family of circuits, we consider that $(h,t,g,z_n) \in \mathrm{DCL(\mathcal{C})}$ iff the gate with number $h$ is of type $t$ and the gate with number $g$ is a child of gate $h$, and $z_n$ is an arbitrary binary string of length $n$, if the type is $\vee$, $\wedge$, or $\neg$. If the type is $x$, then $h$ is an input gate that corresponds to bit $g$ of the input. $\mathrm{bin}(\mathcal{A})$, of length $n$,  is the binary encoding of the input structure $\mathcal{A}$ on which $M_{\mathcal{C}}$ will compute the query (see \cite{Immerman99}). With $\hat{n}$ we denote the size of the domain of $\mathcal{A}$.
  Further, for every $m \geq 1$ we define $qAC^0_{m}$ as the subclass of $qAC^0$ of the families of circuits in $qAC^0$ where the path from  an input gate  to the output gate  with the maximum number of alternated gates of unbounded fan-in of type $\mathrm{AND}$ and $\mathrm{OR}$ in the circuits is $m$.
Following \cite{barrington:jcss1990} we assume that in all the circuits in the family the $\mathrm{NOT}$  gates can occur only at the second level from the left (i.e., immediately following input gates), the gates of unbounded fan-in at any given depth are all of the same type, the layers of such gates alternate in the two types, and the inputs to a gate in a given layer are always outputs of a gate in the previous layer.
Besides the $m$ alternated layers of only gates of unbounded fan-in, we have in each circuit and to the left of those layers a region of the circuit with only $\mathrm{AND}$ and $\mathrm{OR}$ gates of fan-in $2$, with an arbitrary layout, and to the left of that region a layer of some possible $\mathrm{NOT}$ gates and then a layer with the $n$ input gates.
It is straightforward to transform any arbitrary $qAC^0$ circuit into an equivalent one that satisfies such restrictions.
Let us denote as $\exists qAC^0_{m}$ ($\forall qAC^0_{m}$) the subclass of $qAC^0_{m}$ where the output gate is of type $\mathrm{OR}$ ($\mathrm{AND}$).

\begin{lemma}\label{a}
For all $m \ge 1$ we have that $\exists qAC^0_{m}$ $\subseteq$ $\tilde{\Sigma}^{\mathit{plog}}_m$, and $\forall qAC^0_{m}$ $\subseteq$ $\tilde{\Sigma}^{\mathit{plog}}_{m+1}$.


\end{lemma}

\begin{proof}[Proof (sketch)]
Let $\mathcal{C}$ be a family of circuits in $\exists qAC^0_{m}$ with the uniformity conditions given above. We build an alternating Turing machine $M_{\mathcal{C}}$ that computes the query computed by $\mathcal{C}$. We have a deterministic Turing machine that decides the $\mathrm{DCL}$ of $\mathcal{C}$ in time $\log^{c} n$, for some constant $c$. Then  for any given pair of gate numbers $g, h$, gate type $t$, and arbitrary string of $n$ bits $z_{n}$, we can deterministically check both $(h,t,g,z_n) \in \mathrm{DCL(\mathcal{C})}$ and $(h,t,g,z_n) \not\in \mathrm{DCL(\mathcal{C})}$. To be able to have a string $z_n$ of size $n$, we allow the input tape of $M_{\mathcal{C}}$ to be read/write, so that to compute those queries we write   $h,t,g,$ to the left of the input $\mathrm{bin}(\mathcal{A})$ in the input tape of the machine. Note that each gate number is $O(\log^{c'} n)$ bits long for some constant $c'$, and the type $t$ is $2$ bits long (which encodes the type in $\{\mathrm{AND},$ $\mathrm{OR},$ $\mathrm{NOT},$ $x\}$, where $x$ indicates that the gate is an input gate).

First $M_{\mathcal{C}}$ computes the size $n$ of the input $\mathrm{bin}(\mathcal{A})$, which can done in logarithmic time (see \cite{barrington:jcss1990}).

Corresponding to the $m$ alternated layers of  gates of unbounded fan-in of type $\mathrm{AND}$ and $\mathrm{OR}$, where the last layer consists of one single $\mathrm{OR}$ gate which is the output gate, we will have in $M_{\mathcal{C}}$ $m$ alternated blocks of existential and universal states, which will be executed in the opposite direction to the edge relation in the circuit, so that the first block, which is existential will correspond to layer $m$ in $\mathcal{C}$. Each such block  takes time $O(\log^{c'} n)$. In the following we will consider the layers from right to left.
\begin{itemize}
 \item
In the first block, which is \textbf{existential}, $M_{\mathcal{C}}$ guesses gate numbers $g_{o}, h_{2}$, and checks whether $(h_{2}, \wedge, g_{o}, z_n)$ $\in \mathrm{DCL(\mathcal{C})}$. If it is true, then $M_{\mathcal{C}}$ writes in a work tape the sequence $\langle g_{o}, h_{2}\rangle$.
Note that $h_{2}$ is of type $\wedge$. Then it checks
whether $(h_{2}, \wedge, g_{o},  z_n)$ $\not\in \mathrm{DCL(\mathcal{C})}$. If it is true then $M_{\mathcal{C}}$ rejects.

 \item
The second block is \textbf{universal}, and has two stages. In the first stage, $M_{\mathcal{C}}$ checks whether $g_{o}$ is the output gate. To that end it guesses a gate number $u$, and checks whether $(g_{o}, \vee, u,  z_n)$ $\in \mathrm{DCL(\mathcal{C})}$. If it is true, then $M_{\mathcal{C}}$ rejects.

In the second stage, $M_{\mathcal{C}}$ checks the inputs to gate $h_{2}$. It guesses a gate number $h_{3}$, and checks whether $(h_{3}, \vee,  h_{2}, z_n)$ $\in \mathrm{DCL(\mathcal{C})}$, in which case it adds $h_{3}$ at the right end of the sequence in the work tape. It then checks whether $(h_{3}, \vee, h_{2},  z_n)$ $\not\in \mathrm{DCL(\mathcal{C})}$, in which case it accepts.

 \item
The third block is \textbf{existential}. $M_{\mathcal{C}}$ checks the inputs to gate $h_{3}$. It guesses a gate number $h_{4}$, and checks whether
$(h_{4}, \wedge, h_{3}, z_n)$ $\in \mathrm{DCL(\mathcal{C})}$, in which case it adds $h_{4}$ at the right end of the sequence in the work tape. It then checks whether
$(h_{4}, \wedge, h_{3}, z_n)$ $\not\in \mathrm{DCL(\mathcal{C})}$, in which case it rejects.

 \item
Following the same alternating pattern, the $(m-1)$-th block will be existential or universal depending on the type of the gates in $\mathcal{C}$ at the $(m- 1)$-th layer.
Note that, in our progression from the output gate towards the input gates (right to left), the parents of $h_{m}$ (which is the gate number guessed at the $(m-1)$-th block of $M_{\mathcal{C}}$) are the first gates in the region of $\mathcal{C}$ of the arbitrary layout of only gates with bounded fan-in. As the depth of each circuit in the family $\mathcal{C}$ is constant, say it is $w$ for all the circuits in the family, in the $m$-th block $M_{\mathcal{C}}$ can guess the \textit{whole sub-circuit} of that region. Then it guesses $w$ gate  numbers and checks that they form exactly the layout of that region of the circuit. Once  $M_{\mathcal{C}}$ has guessed that layout it can work deterministically to evaluate it, up to the input gates, which takes time $O(1)$.

If the $(m-1)$-th block is \textbf{universal}, then if the guessed $w$ gate numbers do not form the correct layout, $M_{\mathcal{C}}$ accepts.
If the $(m-1)$-th block is \textbf{existential}, and the guessed $w$ gate numbers do not form the correct layout, then $M_{\mathcal{C}}$ rejects.


\end{itemize}


Note that if the family $\mathcal{C}$ is in $\forall qAC^0_{m}$, the first block of states is still existential, to guess the output gate. Then it works as in the  case of $\exists qAC^0_{m}$. 
\end{proof}

\begin{lemma}\label{b}
Let $t,k \ge 1$ and $\psi \in {\Sigma}^{1,\mathit{plog}}_t$ with first-order sub-formula $\varphi \in {\Sigma}^{0}_k$, and whose vocabulary includes the $\mathrm{BIT}$ predicate.
Then there is a family $\mathcal{C}_{\psi}$ of Boolean circuits in $\exists qAC^0_{t+k}$ that computes the Boolean query expressed by $\psi$.

\end{lemma}

\begin{proof}[Proof (sketch)]

We essentially follow the sketch of the proof of the theorem in Section 3, page 89 of \cite{Barr92}, where it shows that for a given $\psi \in $ $\mathrm{SO}^{\mathit{r}}$ there is an equivalent circuit family $\mathcal{C}_{\psi}$ in $qAC^0$. But we use a simple strategy to define a layout of the circuits which will preserve the number of alternated blocks of quantifiers in $\psi$ and $\varphi$.

 Starting from the canonical $AC^0$ circuit corresponding to an first-order formula in prenex, as in Theorem 9.1 in \cite{barrington:jcss1990}, we follow the same idea extending such circuit $\mathcal{C}_{\varphi}$ for $\varphi \in {\Sigma}^{0}_k$, to a canonical circuit $\mathcal{C}_{\psi}$ in $\exists qAC^0$
for $\psi \in {\Sigma}^{1,\mathit{plog}}_t$, in such a way that $\mathcal{C}_{\psi}$ is in $\exists qAC^0_{t+k}$.

The layout of $\mathcal{C}_{\varphi}$ basically follows from left to right the opposite order of the formula $\varphi$, so that the output gate is an unbounded fan-in $\vee$ gate that corresponds to the first-order quantifier $\exists_1$, the inputs to that gate are the outputs of a layer of unbounded fan-in $\wedge$ gates that correspond to the first-order quantifier $\forall_2$, and so on. To the left of the leftmost layer of unbounded fan-in gates corresponding to the quantifier $Q_k$, there is a constant size, constant depth region of the circuit which corresponds to the quantifier free sub-formula of $\varphi$. This part has the input gates, constants, $\mathrm{NOT}$ gates, and  $\mathrm{AND}$ and $\mathrm{OR}$ gates of fan-in $2$.

In a similar way, we extend $\mathcal{C}_{\varphi}$ to the right, to get $\mathcal{C}_{\psi}$. To that end, to the right of the first-order quantifier $\exists_1$ (which in $\mathcal{C}_{\psi}$ becomes a \textit{layer} of unbounded fan-in $\vee$ gates), we will have in $\mathcal{C}_{\psi}$ one layer of  gates of unbounded fan-in for each $SO$ quantifier: of $\vee$ gates for existential $SO$ quantifiers, and $\wedge$ gates for universal $SO$ quantifiers. The first added layer corresponds to the $SO$ quantifier $Q_t$, and the rightmost layer will correspond the $SO$ quantifier $\exists_1$, which is a layer of one single gate, that becomes the new output gate.

Clearly, by following the construction above we get a family of circuits in $\exists qAC^0_{t+k}$.

To build the Turing machine $M_{\mathcal{C}}$ $\in\mathrm{DTIME} [\log^{O(1)} n]\;$ that decides the language $\mathrm{DCL}(\mathcal{C})$, we use the same kind of encoding sketched in \cite{Barr92} (which in turn is an extension of the one used in \cite{barrington:jcss1990} for $AC^{0}$).
In the number of the gates we encode all the information that we need to decide the language, while still keeping its length polylogarithmic.
Each such gate number will have different sections: i) type of the gate; ii) a sequence of $k$ fields of polylogarithmic size each to hold the values of the first-order variables; iii) a sequence of $t$ fields of polylogarithmic size each to hold the values of of the $SO$ variables; iv) a constant size field for the code of the $\mathrm{NOT}$, $\mathrm{OR}$ and $\mathrm{AND}$ gates of fan-in $1, 2, 2$, respectively, in the region of $\mathcal{C}_{\psi}$ corresponding to the quantifier-free part of the formula $\varphi$; v) a logarithmic size field for the bit number that corresponds to an input gate, vi) a polylogarithmic size field for the number of the gate whose output is the left (or only) input to the gate; and vii) idem for the right input.

The idea is that for any given gate, its number will hold the values of all the first-order and $SO$ variables that are bounded in ${\psi}$, in the position of the formula that corresponds, by the construction above, to that gate, or zeroes if the variable is free. Note that each sequence of all those values uniquely define a path in $\mathcal{C}_{\psi}$ from the output gate to the given gate. For the gates which correspond to quantifiers, and for those in the quantifier free part that are parents of them, the encoding allows to easily compute the number of their child gates.
The gates for the quantifier free part hold in their numbers (iv) a number which uniquely identifies that gate in that region of the circuit for a particular branch in $\mathcal{C}_{\psi}$ which is given by the values of the bounded variables. Note that the layout of each such branch of the circuit is constant and hence \textit{stored} in the transition function of $M_{\mathcal{C}}$, and it can evaluate that sub circuit in polylogarithmic time.
Note that the predicate $\mathrm{BIT}(i,j)$ can be evaluated by $M_{\mathcal{C}}$ by counting in binary in a work tape up to $j$ and then looking at its bit $i$.

In this way, clearly $M_{\mathcal{C}}$ decides $\mathrm{DCL}(\mathcal{C})$ in time $\log^{O(1)} n$. 
\end{proof}

\section{The Normal Form for \texorpdfstring{$\mathrm{SO}^{\mathit{plog}}$}{TEXT}}\label{app-snf}

We show the proof of Lemma \ref{lem-snf}

\begin{lemma}[Lemma \ref{lem-snf}]
For every $\mathrm{SO}^{\mathit{plog}}$-formula $\varphi$, there is an equivalent $\mathrm{SO}^{\mathit{plog}}$-formula $\varphi'$ that is in SNF.

\end{lemma} 

\begin{proof}

An easy induction using renaming of variables and equivalences such as $(\neg \exists X^{\log^k} \varphi)$ $\equiv \forall X^{\log^k} (\neg \varphi)$ and $(\phi \vee
\exists x \psi) \equiv \exists x (\phi \vee \psi)$ if $x$ is not free in $\phi$, shows that each $\mathrm{SO}^{\mathit{plog}}$-formula is logically equivalent to an $\mathrm{SO}^{\mathit{plog}}$-formula in \emph{prenex normal form}, i.e., to a formula where all first- and second-order quantifiers are grouped together at the front, forming alternating blocks of consecutive existential or universal quantifiers. Yet the problem is that first- and second-order quantifiers might be mixed. Among the quantifiers of a same block, though, it is clearly possible to commute them 
so as to get those of second-order at the beginning of the block. But, we certainly cannot commute different quantifiers without altering the meaning of the formula. What we can do is to replace first-order quantifiers by second-order quantifiers so that all quantifiers at the beginning of the formula are of second-order, and they are then eventually followed by first-order quantifiers. This can be done using the following equivalences:
\[ \exists x \forall Y^{\log^k} \psi \equiv \exists X^{\log^0} \forall Y^{\log^k} \exists x (X^{\log^0}(x) \wedge \psi).\]
\[ \forall x \exists Y^{\log^k} \psi \equiv \forall X^{\log^0} \exists Y^{\log^k} \forall x (X^{\log^0}(x) \rightarrow \psi). \]
\end{proof}

\section{NODNFSAT Query}\label{appendix:NODNFSAT}

The following sentence expresses the NODNFSAT query described in Section \ref{sec:examples2}. For clarity, we write the query using an unbounded first order universal quantifier. Nevertheless, we can rewrite it as a $\Pi^{\mathit{plog}}_2$ sentence by simply replacing the universal first order quantifiers by a second quantifier as explained in Remark~\ref{r1}.
We use $H$ as a variable with arity $2$ and exponent $2$, and assume the alphabet of input formula in DNF as $\{(,),\wedge,\vee,\neg,0,1,X\}$ with $\sigma=<,I_(,I_),I_\wedge,I_\vee, I_\neg, I_0,I_1,I_X>$, same encoding as in \cite{Cook_71}.

\begin{align*}
    \forall x_b x_c {\Bigg(}_0& \neg I_{(}(x_b) \vee \neg I_{)}(x_c) \vee \exists y \Bigg[_1 y>x_b \wedge y< x_c \wedge (I_{(}(y)\wedge I_{)}(y))\Bigg]_1 \vee\\
    & \exists H\exists x_ax'_ax_{aa}x'_{aa}x_fx'_f\forall xx'yy'\Bigg[_1 \bigg(_2 \Big(_3 H(x,x')\wedge H(y,y')\Big)_3\to\\
    &\bigg[_3\Big(_4 (x_b \leq x \wedge x \leq x_c) \wedge (x_b \leq x' \wedge x' \leq x_c) \wedge (x_b \leq y \wedge y \leq x_c) \wedge\\
    &\text {\ \ \ \ \ } 
    (x_b \leq y' \wedge y' \leq x_c)\Big)_4 \wedge
    \Big(_4 x=y{\to} x'=y'\Big)_4 \wedge
    \Big(_4 x'=y'{\to} x=y\Big)_4 \\
    &
    \wedge \Big(_4 
    x=x_a\vee x=x_{aa} \vee x=x_f\vee \big(_5 I_0(x)\wedge I_0(x')\big)_5
    \Big)_4\\
    &
    \wedge \Big(_4 
    x'=x'_a\vee x'=x'_{aa} \vee x'=x'_f\vee \big(_5 I_1(x)\wedge I_1(x')\big)_5
    \Big)_4\\
    &\wedge \Big(_4 SUCC(x,y){\leftrightarrow} SUCC(x',y')\Big)_4\wedge\\
    &
    \Big(_4
    \forall zz'\big(_5 (_6 H(z,z')\wedge(I_0(z)\vee I_1(z))\wedge(I_0(x)\vee I_1(x)))_6{\to} x\leq z\big)_5 \leftrightarrow\\
    &
    \forall vv'\big(_5 (_6 H(v,v')\wedge(I_0(v')\vee I_1(v'))\wedge(I_0(x')\vee I_1(x')))_6{\to} x'\leq v'\big)_5 
    \Big)_4 \wedge\\
    &\Big[_4 
    \big[_5 
        \forall zz'\big(_5 (_6 H(z,z')\wedge(I_0(z)\vee I_1(z))\wedge(I_0(x)\vee I_1(x)))_6{\to} x\leq z\big)_5 \wedge\\
        &
        SUCC(x_a,x){\wedge}SUCC(x_{aa},x_a){\wedge}SUCC(x'_a,x'){\wedge} SUCC(x'_{aa},x'_a)
    \big]_5\to\\
    &
    \big[_5
    I_X(x_a)\wedge I_X(x'_a)\wedge \big(_5 I_((x_{aa})\vee I_\wedge(x_{aa})\vee I_\neg(x_{aa})\big)_5 \wedge \big(_5 I_((x'_{aa})\vee \\
    &\text{\ \ }
    I_\wedge(x'_{aa}) \vee I_\neg(x'_{aa})\big)_5 \wedge \big(_5 I_\neg(x_{aa})\leftrightarrow I_\wedge(x'_{aa})\big)_5
    \big]_5 \Big]_4 \wedge\\
    &
    \Big(_4
    \forall zz'\big(_5 (_6 H(z,z')\wedge(I_0(z)\vee I_1(z))\wedge(I_0(x)\vee I_1(x)))_6{\to} z\leq x\big)_5 \leftrightarrow\\
    &
    \forall vv'\big(_5 (_6 H(v,v')\wedge(I_0(v')\vee I_1(v'))\wedge(I_0(x')\vee I_1(x')))_6{\to} v'\leq x'\big)_5 
    \Big)_4 \wedge\\
    &\Big[_4 
    \big[_5 
        \forall zz'\big(_5 (_6 H(z,z')\wedge(I_0(z)\vee I_1(z))\wedge(I_0(x)\vee I_1(x)))_6{\to} z\leq x\big)_5 \wedge\\
        &
        SUCC(x,x_f)\wedge SUCC(x',x'_f)
    \big]_5\to
    \big[_5
        \big(_5 I_)(x_f)\vee I_\wedge(x_f)\big)_5\wedge
        \big(_5 I_)(x'_f)\vee \\
        &
        \text{\ \ \ }
        I_\wedge(x'_f)) \big)_5
    \big]_5 \Big]_4 \bigg]_3 \bigg)_2 \Bigg]_1 {\Bigg)}_0
\end{align*}

\section{Details of Bounded Binary Arithmetic in \texorpdfstring{$\mathrm{SO}^{\mathit{plog}}$}{TEXT} }\label{app-arithmetics}

We describe next the sub-formulae $\varphi$, $\psi$ and~$\alpha$ of $\mathrm{BSUM}_k$ (please, refer to formula~(\ref{arith8}) in the main text).

The sub-formulae $\varphi(X, Y, Z)$ is satisfied if the value of the least significant bit of $Z$ is consistent with the sum of the least significant bits of $X$ and $Y$.  
\begin{align*}
\varphi(X, Y, Z) \equiv &
\big( Z(\bar{0},0) \wedge ( (X(\bar{0},0) \wedge Y(\bar{0},0)) \vee (X(\bar{0},1) \wedge Y(\bar{0},1)) ) \big) \vee\\
&\big( Z(\bar{0},1) \wedge ( (X(\bar{0},1) \wedge Y(\bar{0},0)) \vee (X(\bar{0},0) \wedge Y(\bar{0},1)) ) \big)
\end{align*} 

The sub-formulae $\psi(\bar{x}, \bar{y}, W, X, Y)$ is satisfied if the value of the bit in position $\bar{x}$ of $W$ (i.e., the value of the carried bit), is consistent with the sum of the values of the bits in position $\bar{y}$ (i.e., the position preceding $\bar{x}$) of $W$, $X$ and $Y$.  
\begin{align*}
\psi(&\bar{x}, \bar{y}, W, X, Y) \equiv \\
&\big( W(\bar{x},0) \wedge ( (W(\bar{y},0)\wedge X(\bar{y},0)\wedge Y(\bar{y},0)) \vee (W(\bar{y},0)\wedge X(\bar{y},0)\wedge Y(\bar{y},1))\vee\\
&\hspace*{1.9cm} (W(\bar{y},0)\wedge X(\bar{y},1)\wedge Y(\bar{y},0)) \vee (W(\bar{y},1)\wedge X(\bar{y},0)\wedge Y(\bar{y},0))) \big) \vee \\
&\big( W(\bar{x},1) \wedge ( (W(\bar{y},1)\wedge X(\bar{y},1)\wedge Y(\bar{y},0)) \vee (W(\bar{y},1)\wedge X(\bar{y},0)\wedge Y(\bar{y},1))\vee\\
&\hspace*{1.9cm} (W(\bar{y},0)\wedge X(\bar{y},1)\wedge Y(\bar{y},1)) \vee (W(\bar{y},1)\wedge X(\bar{y},1)\wedge Y(\bar{y},1))) \big) 
\end{align*} 

Finally, $\alpha(\bar{x}, W, X, Y, Z)$ is satisfied if the value of the bit in position $\bar{x}$ of $Z$ is consistent with the sum of the corresponding bit values of $W$, $X$ and $Z$.  
\begin{align*}
\alpha(&\bar{x}, W, X, Y, Z) \equiv \\
&\big( Z(\bar{x},0) \wedge ( (W(\bar{x},0)\wedge X(\bar{x},0)\wedge Y(\bar{x},0)) \vee (W(\bar{x},0)\wedge X(\bar{x},1)\wedge Y(\bar{x},1))\vee\\
&\hspace*{1.8cm} (W(\bar{x},1)\wedge X(\bar{x},1)\wedge Y(\bar{x},0)) \vee (W(\bar{x},1)\wedge X(\bar{x},0)\wedge Y(\bar{x},1))) \big) \vee \\
&\big( Z(\bar{x},1) \wedge ( (W(\bar{x},0)\wedge X(\bar{x},0)\wedge Y(\bar{x},1)) \vee (W(\bar{x},0)\wedge X(\bar{x},1)\wedge Y(\bar{x},0))\vee\\
&\hspace*{1.8cm} (W(\bar{x},1)\wedge X(\bar{x},0)\wedge Y(\bar{x},0)) \vee (W(\bar{x},1)\wedge X(\bar{x},1)\wedge Y(\bar{x},1))) \big)
\end{align*}

In the following formula expressing $\mathrm{BMULT}_k(X,Y,Z)$, the variable $I$ has arity $k$ and exponent $k$. $I'$ is of arity $2k$ and exponent $2k$. The remaining second-order variables $R$, $S$ and $W$ are of arity $2k + 1$ and exponent $2k$.
\begin{align}
\mathrm{BMULT}_k & (X,Y,Z) \equiv \notag \\
\exists I I' & R S W \big(\mathrm{DEF}_k(I) \wedge \mathrm{BIN}_k(X,I) \wedge \mathrm{BIN}_k(Y,I) \wedge \mathrm{BIN}_k(Z,I) \wedge \notag \\
&\mathrm{DEF}_{2k}(I') \wedge \mathrm{BIN}_{2k}(R,I') \wedge \mathrm{BIN}_{2k}(S,I') \wedge \mathrm{BIN}_{2k}(W,I') \wedge \notag \\
&\mathrm{SHIFT}(S, X, I) \wedge \notag \\
&\forall \bar{x} (I(\bar{x}) \to ( (\bar{x} = \bar{0} \wedge Y(\bar{x},0) \wedge \varphi_a(R, \bar{x})) \vee \notag \\
&\hspace*{1.5cm} (\bar{x} = \bar{0} \wedge Y(\bar{x},1) \wedge \varphi_b(R, \bar{x}, X)) \vee \notag \\
&\hspace*{1.5cm} (Y(\bar{x},0) \wedge \exists \bar{y} (SUCC_k(\bar{y},\bar{x}) \wedge \varphi_c(R, \bar{x}, \bar{y}))) \vee \notag \\
&\hspace*{1.5cm} (Y(\bar{x},1) \wedge \exists \bar{y} (SUCC_k(\bar{y},\bar{x}) \wedge \varphi_d(R, S, W, \bar{x}, \bar{y})))))\big) \label{arith9}
\end{align}
The sub-formula $\varphi_a(R, \bar{x}) \equiv \forall \bar{y} (I(\bar{y}) \to R(\bar{x}, \bar{y}, 0))$  expresses that $R|_{\bar{x}}$ encodes the binary number $0$, the sub-formula $\varphi_b(R, \bar{x}, X) \equiv \forall \bar{y} (I(\bar{y}) \to \exists z(R(\bar{x}, \bar{y}, z) \wedge X(\bar{y}, z)))$ expresses that $R|_{\bar{x}} = X$, the sub-formula $\varphi_c(R, \bar{x}, \bar{y}) \equiv \forall \bar{w} (I(\bar{w}) \to \exists z(R(\bar{x}, \bar{w}, z) \wedge R(\bar{y}, \bar{w}, z)))$ expresses that $R|_{\bar{x}} = R|_{\bar{y}}$, and the sub-formula $\mathrm{SHIFT}(S, X, I)$ expresses that if $\bar{a} \in B^k$ is the $i$-th tuple in the numerical order of $B^k$, then $S|_{\bar{a}}$ is the $(i-1)$-bits arithmetic left-shift of $X$, i.e., $S|_{\bar{a}}$ is $X$ multiplied by $2^{i-1}$ in binary. We have
\begin{align}
\mathrm{SHIFT}(S, X&, I) \equiv \exists \bar{x} \bar{y} \big(\mathrm{SUCC}_k(\bar{x}, \overline{\mathit{logn}}) \wedge \mathrm{SUCC}_k(\bar{y}, \bar{x}) \wedge S(\bar{y}, \bar{x}, 0) \big) \notag \\
\forall \bar{x} \big(I(\bar{x}) \to &( (\bar{x} = \bar{0} \wedge \varphi_b(S,\bar{x}, X)) \vee \notag \\
&\hspace*{0.2cm} \exists \bar{y} (\mathrm{SUCC}_k(\bar{y},\bar{x}) \wedge \notag \\
&\hspace*{0.8cm} \forall \bar{z} (I(\bar{z}) \to ((\bar{z} = \bar{0} \wedge S(\bar{x}, \bar{z}, 0)) \vee \notag \\
&\hspace*{1.7cm} \exists \bar{z}' b (\mathrm{SUCC}_k(\bar{z}',\bar{z}) \wedge S(\bar{y}, \bar{z}', b) \wedge S(\bar{x}, \bar{z}, b))))))\big)   \label{arith10}  
\end{align}
Finally, the sub-formula $\varphi_d(R, S, W, \bar{x}, \bar{y})$ expresses that $R|_{\bar{x}}$ results from adding $R|_{\bar{y}}$ to $S|_{\bar{x}}$. The carried digits of this sum are are kept in $W|_{\bar{x}}$. Given the formula $\mathrm{BSUM}_k$ described earlier, it is a straightforward task to write $\varphi_d$. We omit further details.  

$\mathrm{BDIV}_k(X, Y, Z, M)$ can be written as follows.
 \begin{align}
    \mathrm{BDIV}_k(X, Y, Z, M)\equiv \notag \\
\exists I I' A R S W W' \big(&\mathrm{DEF}_k(I) \wedge \mathrm{BIN}_k(X,I) \wedge \mathrm{BIN}_k(Y,I) \wedge \mathrm{BIN}_k(Z,I) \wedge \notag \\
    &\mathrm{BIN}_k(M,I) \wedge \mathrm{BIN}_k(A,I) \wedge \neg \mathrm{BNUM}_k(Y,0,I) \wedge \notag \\
    &{<_k}(M, Y, I) \wedge \mathrm{BMULT}_k(Z, Y, A, I, I', R, S, W) \wedge \notag \\
    &\mathrm{BSUM}_k(A,M,X, I, W')\big). \label{arith11}
\end{align}

Where in the previous formula $\mathrm{BMULT}_k(Z, Y, A, I, I', R, S, W)$ denotes the formula obtained from $\mathrm{BMULT}_k(X, Y, Z)$ by eliminating the second-order quantifiers (so that $I$, $I'$, $R$, $S$ and $W$ become free-variables) and by renaming $X$ and $Z$ as $Z$ and $A$, respectively. Likewise, $\mathrm{BSUM}_k(A,M,X, I, W')$ denotes the formula obtained from $\mathrm{BSUM}_k(X, Y, Z)$ by eliminating the second-order quantifiers (so that $I$ and $W$ become free-variables) and by renaming $X$, $Y$, $Z$ and $W$ as $A$, $M$, $X$ and $W'$, respectively.

\section{Sketches of the Proofs of Theorems \ref{pedsoplog2} and \ref{pedsoplog3}}\label{app-plh}

\begin{theorem}[Theorem \ref{pedsoplog2}]
Over ordered structures with sucessor relation, $\mathrm{BIT}$ and constants for $\log n$, the minimum, second and maximum elements, $\Pi^{\mathit{plog}}_1$ captures $\tilde{\Pi}^{\mathit{plog}}_1$.
\end{theorem}

\begin{proof}[Proof (sketch)]
\ In order to show $\Pi^{\mathit{plog}}_1 \subseteq \tilde{\Pi}^{\mathit{plog}}_1$ we proceed in the same way as in the proof of Theorem \ref{pedsoplog}[Part a] with the only difference that all states are universal. Let $\phi = \forall X_1^{r_1,\log^{k_1}} \dots \exists X_m^{r_m,\log^{k_m}} \varphi$, where $\varphi$ is a first-order formula with the restrictions given in the definition of $\mathrm{SO}^{\mathit{plog}}$. We first determine all possible values for the second-order variables $X_i^{r_i,\log^{k_i}}$. Any combination of such values determines a branch in the computation tree of \textbf{M}, and for each such branch the machine has to checks $\varphi$. The argument that these checks can be done in poly-logarithmic time is the same as in the proof of Theorem \ref{pedsoplog}. Then by definition of the complexity classes $\mathrm{ATIME}^{op}[\log^k n,m]$ and the definition of acceptance for alternating Turing machines the machine \textbf{M} evaluates $\phi$ in poly-logarithmic time.

In order to show the inverse, i.e. $\tilde{\Pi}^{\mathit{plog}}_1 \subseteq \Pi^{\mathit{plog}}_1$, we exploit that the given random access alternating Turing machine has only universal states and thus all branches in its computation tree must lead to an accepting state. Consequently, the same construction of a formula $\phi$ as in the proof of Theorem \ref{pedsoplog}[Part b] can be used with the only difference that all second-order existential quantifiers have to be turned into universal ones. Then the result follows in the same way as in the proof of Theorem \ref{pedsoplog}.
\end{proof}

\begin{theorem}[Theorem \ref{pedsoplog3}]
Over ordered structures with sucessor relation, $\mathrm{BIT}$ and constants for $\log n$, the minimum, second and maximum elements, $\Sigma^{\mathit{plog}}_m$ captures $\tilde{\Sigma}^{\mathit{plog}}_m$ and $\Pi^{\mathit{plog}}_m$ captures $\tilde{\Pi}^{\mathit{plog}}_m$ for all $m \ge 1$.
\end{theorem}

\begin{proof}[Proof (sketch)]
\ We proceed by induction, where the grounding cases for $m=1$ are given by Theorems \ref{pedsoplog} and \ref{pedsoplog2}. For the inclusions $\Sigma^{\mathit{plog}}_m \subseteq \tilde{\Sigma}^{\mathit{plog}}_m$ and $\Pi^{\mathit{plog}}_m \subseteq \tilde{\Pi}^{\mathit{plog}}_m$ we have to guess (or take all) values for the second-order variables in the leading block of existential (or universal, respectively) quantifiers, which is done with existential (or universal, respectively) states. For the checking of the subformula in $\Pi^{\mathit{plog}}_{m-1}$ (or in $\Sigma^{\mathit{plog}}_{m-1}$, respectively) we have to switch to a universal (existential) state and apply the induction hypothesis for $m-1$.

Conversely, we consider the computation tree of the given alternating Turing machine \textbf{M} and construct a formulae as in the proofs of Theorems \ref{pedsoplog} and \ref{pedsoplog2} exploiting that for each switch of state from existential to universal (or the other way round) the corresponding submachine can by induction be characterised by a formula in $\Pi^{\mathit{plog}}_{m-1}$ or $\Sigma^{\mathit{plog}}_{m-1}$, respectively.
\end{proof}

\end{appendix}

\end{document}